\documentclass[journal,twoside,web]{ieeecolor}

\makeatletter
\def\ps@titlepagestyle{%
  \def\@oddfoot{}%
  \def\@evenfoot{}%
  \if@confmode
    \def\@oddhead{}%
    \def\@evenhead{}%
  \else
    \def\@oddhead{\vbox{\vbox{}\hskip-3pc\null\hskip\logowidth
      \vbox{\begin{tabular}{@{\hspace*{37pt}}p{\firstpagerule}@{}}\\[-20pt]{\vbox{\hsize\firstpagerule\scriptsize\textsf\leftmark \hfill \textsf\thepage\hbox{}\par\vspace*{-3pt}
      \vbox{\color{subsectioncolor}\hrule height1pt width\firstpagerule depth0pt} }}\end{tabular}}}}%
    \def\@evenhead{\begin{tabular}{@{}p{\firstpagerule}@{}}\\[-38pt]{\vbox{\hsize\firstpagerule\scriptsize\textsf\thepage \hfill \textsf\leftmark\hbox{}\par\vspace*{-3pt} 
      \vbox{\color{subsectioncolor}\hrule height1pt width\firstpagerule depth0pt} }}\end{tabular}\hskip-5pc\null\hskip\logowidth}%
  \fi
}
\makeatother

\usepackage{generic}
\usepackage{amsmath,amssymb,amsfonts}
\usepackage{algorithmic}
\usepackage{graphicx}
\usepackage{booktabs} 
\usepackage{multirow}
\usepackage{hyperref}
\hypersetup{hidelinks}
\usepackage{textcomp}
\usepackage[mathscr]{euscript} 
\usepackage{mathtools}
\newcommand\mycom[2]{\genfrac{}{}{0pt}{}{#1}{#2}}

\usepackage[linesnumbered,ruled,vlined]{algorithm2e}
\usepackage{array}

\makeatletter
\ExplSyntaxOn
\NewDocumentCommand { \rightdiamond } { O{} m }
  {
    \mode_if_math:TF
      { \__examzh_cdotfill: \tag*{$\Diamond$} }
      { \__examzh_cdotfill:  }
    \mode_if_math:F
      {
        \par \noindent \ignorespaces
      }
  }
{
\cs_new:Npn \__examzh_cdotfill:
  {
    \mode_leave_vertical:
    \cleaders \hb@xt@ .44em {\hss $\cdot$ \hss} \hfill
    \kern\z@
  }
}
\ExplSyntaxOff
\makeatother

\newtheorem{definition}{Definition}
\newtheorem{theorem}{Theorem}
\newtheorem{corollary}{Corollary}

\newtheorem{example}{Example}
\newtheorem{proposition}{Proposition}
\newtheorem{remark}{Remark}

\def\BibTeX{{\rm B\kern-.05em{\sc i\kern-.025em b}\kern-.08em
    T\kern-.1667em\lower.7ex\hbox{E}\kern-.125emX}}
\begin{document}
\title{Verification of \textit{K}- and Infinite-Step Strong/Weak Anonymity Using Concurrent Compositions}
\author{Jiahui Zhang, Kuize Zhang, \IEEEmembership{Senior Member, IEEE}, Xiaoguang Han, \IEEEmembership{Member, IEEE},\\ and Zhiwu Li, \IEEEmembership{Fellow, IEEE}
\thanks{This work is partially supported by the National Natural Science Foundation of China under Grand No. 61903274 and the Science and Technology Fund, FDCT, Macau SAR under Grant No. 0029/2023/RIA1. \textit{(Corresponding author: Xiaoguang Han.)}}
\thanks{Jiahui Zhang and Zhiwu Li are with the Institute of Systems Engineering, Macau University of Science and Technology, Taipa 999078, China (e-mail: 3220006067@student.must.edu.mo; zwli@must.edu.mo).}
\thanks{Kuize Zhang is with School of Mathematics and Statistics, Xi’an Jiaotong University, Xi’an 710049, China (e-mail: kuize.zhang@unica.it).}
\thanks{Xiaoguang Han is with the College of Electronic Information and Automation, Tianjin University of Science and Technology, Tianjin 300457, China (e-mail: xiaoguanghan@tust.edu.cn).}}

\maketitle

\begin{abstract}
Anonymity is an information flow property that provides privacy protection in the sense of non-uniqueness of system information at certain moments with respect to observations.
The notion of \textit{K}-step anonymity in the context of discrete-event systems characterizes the scenario that the state estimates cannot be a singleton within at most \textit{K} observational steps prior to the current instant, while infinite-step anonymity is the same as \textit{K}-step anonymity without considering the limit on \textit{K}.
In this paper, we lucubrate \textit{K}- and infinite-step anonymity for partially-observed discrete-event systems modeled by non-deterministic finite-state automata.
First, we define two strong types and two weak types of \textit{K}- and infinite-step anonymity that are fundamentally different from the existing notions of \textit{K}- and infinite-step anonymity due to the consideration of strong and weak anonymous projections.
Then, we develop a new methodology by exploiting the concurrent-composition technique to verify these four types of anonymity.
Based on the constructed concurrent compositions, verifiable necessary and sufficient conditions for the four types of anonymity are provided, along with their complexity analysis.
Finally, the upper bounds on \textit{K} for \textit{K}-step strong anonymity and weak anonymity are computed.
\end{abstract}


\section{Introduction}\label{sec:introduction}
\IEEEPARstart{S}{ecurity} and privacy of cyber-physical systems (CPSs) have increasingly received attention.
As an information flow property, anonymity characterizes the non-uniqueness of the intruder's estimate of a system's state.
Anonymity can be used in a wide range of practical scenarios such as anonymous voting, donations, and transactions \cite{Schneider1996}.

In recent decades, much work has been done on anonymity analysis and verification, see, e.g., \cite{Reiter1998}--\cite{Bryans2008}.
In \cite{Reiter1998}, the authors introduce a system called the Crowds to guarantee the anonymity of users on the internet.
For anonymous Web browsing of the Crowds system, a probabilistic model checker is proposed \cite{Shmatikov2004}. 
The authors in \cite{Dingledine2003} touch upon anonymity and reputation with a focus on anonymous remailers and anonymous publishing, and explain why systems can benefit from reputation.
The work in \cite{Bryans2008} introduces the notions of strong and weak $\mathcal{O}$-anonymity in labeled transition systems to capture different concealment requirements depending on the different observational capabilities of intruders. 
In recent years, anonymity has been extensively studied in discrete-event systems (DESs).
Specifically, the concepts of strong and weak anonymity are proposed in \cite{Feng2011} for describing the ability to hide an action among a set of actions, which are quite different from those in \cite{Bryans2008}.
The notions of anonymity defined in \cite{Bryans2008} and \cite{Feng2011} are incomparable.
To capture the inability of intruders to determine the accurate current state of a system, current-state anonymity is studied in finite-state automata \cite{Wu2014}.
Naturally, the notions of $K$-step anonymity and infinite-step anonymity are investigated, which generalizes the current instant to the last $K$ and arbitrary instants compared with the current-state anonymity \cite{Basilio2021}.
As a stronger version of $K$-anonymity \cite{Sweeney2002} in the computer science area, a notion of infinite-step $K$-anonymity is defined in \cite{Yin2020}.

The underlying observation mechanism is the key to analyzing observational properties of DESs, e.g., anonymity \cite{Bryans2008}--\cite{Yin2020}, opacity \cite{Saboori2007}--\cite{Yang2023}, detectability \cite{Shu2007}--\cite{Zhou2020}, and diagnosability \cite{Lafortune2018}--\cite{Ma2023}.
To characterize different capabilities for observing a system, the work in \cite{Bryans2008} first systematically summarizes three types of observation functions: static observation functions, dynamic observation functions \cite{Zhang2015}, and Orwellian-type observation functions \cite{Hou2022}.
In this paper, we formulate two types of observation functions: strong and weak anonymous projections to define interesting and useful types of anonymity.
Specifically, the strong anonymous projection is a static observation function that describes the scenarios in which an intruder's observation of some specific actions performed by different agents is the same, preventing the intruder from determining the identity of the agents.
In contrast, the weak anonymous projection is a dynamic observation function that characterizes the scenarios in which an intruder is able to recognize certain actions performed only by the same agent without knowing the agent's identity.
In this case, the observation of any action depends on its prefix that has happened.

Recently, a unified concurrent-composition method has been developed for the verification problem of state/event inference and concealment for DESs modeled by finite-state automata \cite{Kuize2023}.
In addition, standard and strong current-state opacity, initial-state opacity, $K$-step opacity, and infinite-step opacity can be verified by using a combination of concurrent composition and the classical observer \cite{Kuize2023}--\cite{Xiaoguang2022}.
The concurrent-composition technique shows its efficiency in verifying these properties in DESs compared with the existing methods.

Inspired by the above work, this paper mainly focuses on the verification of $K$- and infinite-step anonymity under strong and weak anonymous projections by exploiting the concurrent-composition technique.
Specifically, the main contributions of this research are as follows.
\begin{itemize}
\item To characterize different inabilities of an intruder to infer whether a system of interest is in an exact state at some certain instants, we formulate four notions of $K$- and infinite-step anonymity under strong and weak anonymous projections that are essentially different from the existing notions of $K$- and infinite-step anonymity \cite{Basilio2021}, \cite{Yin2020}.
\item A general methodology for verifying these four types of anonymity is developed by constructing two different information structures $Cc(G_m, Obs_m(G))$ (see Eq.~(\ref{eq:5})) and $Cc(\widehat{G}_w, Obs_w(\widehat{G}))$ (see Eq.~(\ref{eq:9})), both of which are concurrent compositions of a particular automaton and a particular observer, where $Cc(G_m, Obs_m(G))$ (resp., $Cc(\widehat{G}_w, Obs_w(\widehat{G}))$) is used to verify $K$- and infinite-step anonymity under the strong (resp., weak) anonymous projection.
\item 
Based on the constructed concurrent compositions, verifiable necessary and sufficient conditions are provided to verify these four types of anonymity separately.
In addition, we provide an algorithm for computing upper bounds on $K$ for $K$-step strong anonymity and weak anonymity.
Finally, an anonymous questioning process in chemical engineering is shown as a case study to illustrate the proposed notions of anonymity.
\end{itemize}

\textbf{Structure.}
The remainder of this paper is structured as follows.
The preliminaries used throughout this paper are presented in Section \ref{sec:2}.
Section \ref{sec:3} formulates four definitions of $K$- and infinite-step anonymity under strong and weak anonymous projections.
Section \ref{sec:4} develops a methodology for verifying the proposed four types of anonymity by exploiting the concurrent composition technique.
We explore the maximal value of $K$ and calculate the upper bounds on $K$ in Section \ref{sec:5} regarding the two types of $K$-step anonymity separately.
Section \ref{sec:6} provides a case study on an anonymous questioning process.
Finally, we conclude this paper in Section \ref{sec:7}.


\section{Preliminaries}\label{sec:2}
\subsection{Notation}\label{sec:2.1}
In this paper, let $\mathbb{N}$ be the set of non-negative integers.
We denote by $P$, $P_{str}$, and $P_w$ the natural projection, the strong anonymous projection, and the weak anonymous projection, respectively.
The set of events is symbolized by $\Sigma$.
Notation $\Sigma^{\ast}$ stands for the set of all finite strings over $\Sigma$, including the empty string $\epsilon$.
Given a string $s = \sigma_1\sigma_2 \ldots \sigma_n$, we use $|s|$ to denote its length, i.e., $|s|= n $ and $|\epsilon| = 0$.
Notation $\Gamma$ is used to represent the set of anonymous events.
Let $\gamma$ be the unique observation of all anonymous events through the strong anonymous projection $P_{str}$.
We denote by $\Gamma_w=\{\gamma_1, \gamma_2, \dots, \gamma_{|\Gamma|}\}$ the set of distinct observation labels of anonymous events in all strings generated by a system under the weak anonymous projection $P_w$, where $|\Gamma|$ is the number of the anonymous events in $\Gamma$.
Given a state set $Q$, $2^Q$ and $UR(Q)$ signify the power set of $Q$ and the unobservable reach of $Q$, respectively.
We use notation $\widehat{E}_{G}(P_{str}, \alpha|\alpha\beta)$ (resp., $\widehat{E}_{G}(P_w, \alpha|\alpha\beta)$) to represent the \emph{delayed state estimate} of system $G$ under the \emph{strong} (resp., \emph{weak}) \emph{anonymous projection} $P_{str}$ (resp., $P_w$) at the instant that $\alpha$ has just been observed upon observing $\alpha\beta$.
In addition, the concurrent composition of automata $G^1$ and $G^2$ is written as $Cc(G^1, G^2)$.

\subsection{System model}\label{sec:2.2}
A non-deterministic finite-state automaton (NFA) is a quadruple
\begin{equation}\label{eq:1}
G = (X, \Sigma, \delta, X_0),
\end{equation}
where
\begin{itemize}
\item{$X$ is a finite set of states,}
\item{$\Sigma$ is a finite set of events,}
\item{$\delta: X\times \Sigma\rightarrow 2^X$ is the transition function, which characterizes the dynamics of $G$: $y\in \delta(x,\sigma)$ denotes that there exists a transition labeled by event $\sigma$ reaching state $y$ from state $x$, and}
\item{$X_0\subseteq X$ is a set of initial states.}
\end{itemize}
The transition function can be extended to $\delta: X\times \Sigma^{\ast}\rightarrow 2^X$ in the usual manner. 
$\mathcal{L}(G,x)$ is used to denote the language generated by $G$ from the state $x$, i.e., $\mathcal{L}(G,x)=\{s\in \Sigma^*|\delta(x,s)\neq \emptyset\}$.
The set $\mathcal{L}(G)=\{s\in \Sigma^*|\exists x_0\in X_0: \delta(x_0, s)\neq \emptyset\}$ denotes the language generated by $G$.
Note that $G$ is called a deterministic finite-state automaton (DFA) if $|X_0|=1$ and $|\delta(x,\sigma)|\leq 1$ for all $x\in X$ and all $\sigma\in \Sigma$.

In general, a system modeled by $G$ is partially observable.
The event set $\Sigma$ is divided into two disjoint subsets $\Sigma_{o}$ and $\Sigma_{uo}$, i.e., $\Sigma = \Sigma_{o} \dot\cup \Sigma_{uo}$, where $\Sigma_o$ is the set of observable events and $\Sigma_{uo}$ is the set of unobservable events.
The natural projection 
$P: \Sigma^*\rightarrow \Sigma_o^*$ is recursively defined by $P(\epsilon)=\epsilon \mbox{ and }$
\begin{equation*}
P(s\sigma)=
\begin{cases}
P(s)\sigma,&\mbox{if }\sigma\in \Sigma_o,\\
P(s),&\mbox{otherwise}.
\end{cases}
\end{equation*}
where $s\in \Sigma^* \ \mbox{and} \ \sigma\in \Sigma$. 
Accordingly, the projection of language generated by $G$, denoted by $P(\mathcal{L}(G))$, is $P(\mathcal{L}(G))=\{P(s)\mid s\in \mathcal{L}(G)\}$.
We refer the reader to \cite{Cassandras2008} for details.

Since system $G$ is partially observed, an intruder needs to estimate the state reached by $G$ based on an observed string, which involves state estimation problems.
The following type of state estimate is called the delayed state estimate under the natural projection $P$ \cite{Cassandras2008}.

\begin{definition}\label{def:2.1}(Delayed state estimate)
Let $\alpha\beta$ be an observation string generated by a system $G = (X, \Sigma, \delta, X_0)$ under the natural projection $P$, i.e., $\alpha\beta\in P(\mathcal{L}(G))$.
The \emph{delayed state estimate} under the natural projection $P$ for the instant of $\alpha$ upon the fact that $\alpha\beta$ is exactly observed is the set of states (at which the system could be in) $|\beta|$ steps ahead, i.e.,
    \begin{equation*}
    \widehat{E}_{G}(P, \alpha|\alpha\beta)=
    \begin{Bmatrix}
    x\in X\Bigg|                     
    &
    \begin{matrix}
    [\exists x_0\in X_0, \exists s\in \mathcal{L}(G, x_0):\\
    P(s)=\alpha \land x\in \delta(x_0, s)] \land\\
    [\exists t\in \mathcal{L}(G, x): P(st)=\alpha\beta] 
    \end{matrix}               
    \end{Bmatrix}.\rightdiamond{}
    \end{equation*}
\end{definition}
By definition, the \emph{delayed state estimate} is the collection of possible states that system $G$ may visit prior to some steps of the current observation. 
The cardinality of $\widehat{E}_{G}(P, \alpha|\alpha\beta)$ is denoted as $|\widehat{E}_{G}(P, \alpha|\alpha\beta)|$, indicating the number of the states in $\widehat{E}_{G}(P, \alpha|\alpha\beta)$.
In addition, we define the unobservable reach of a state subset $Q\subseteq X$ as
\begin{equation*}
UR(Q)=\{q'\in X|\exists q\in Q, \exists s\in \Sigma_{uo}^*: q'\in \delta(q, s)\}.
\end{equation*}

\section{Four types of K- and infinite-step anonymity under the strong/weak anonymous projection}\label{sec:3}
\subsection{K- and infinite-step strong anonymity}
In this paper, we assume that an action (modeled by an event $\sigma\in \Sigma$) is observed (or detected) by an intruder if a sensor is deployed \cite{Hadjicostis2020}.
On the other hand, if no sensor is deployed, an intruder will never observe or detect the occurrence of any action.
An action that can (resp., cannot) be observed by an intruder is called an observable (resp., unobservable) event.
As mentioned in \cite{Schneider1996}, the persons who perform actions, such as anonymous voting, donations, and transactions, are called agents.
In situations where anonymous operations are required, e.g., the aforementioned anonymous transactions, some agents may want to hide their identities.
In this paper, an action that can be observed by an intruder and taken by an agent who wants to conceal his/her identity is modeled as an anonymous event.
We set a subset of observable events, denoted by $\Gamma\subseteq \Sigma_o$, as the anonymous event set.
When an agent performs the action that is modeled as an anonymous event, he/she does not want to be distinguished from other agents who perform actions in some sense.
We assume $|\Gamma|\geq 2$ to make it meaningful to conceal at least two agents' identities.

Suppose that anonymous processors are deployed on the sensors related to the actions modeled by anonymous events, which can help agents hide their identities. 
In this case, an intruder's observations of these actions, modeled by different anonymous events and performed by different agents, are identical and indistinguishable.
We use a new symbol $\gamma\notin \Sigma$ to denote the identical observation obtained by the intruder.
On this basis, an ideal anonymous processor is modeled as the strong anonymous projection as follows.
\begin{definition}\label{def:3.1-2}
    An ideal anonymous processor, i.e., a processor whose processing results for all anonymous events are the same, and whose processing result for a non-anonymous observable event is still the event itself, is modeled as the \emph{strong anonymous projection} $P_{str}: \Sigma^*\rightarrow (\Sigma_o\cup \{\gamma\})^*$ (a static observation function) defined by $P_{str}(\epsilon)=\epsilon \mbox{ and }$
\begin{equation*}
P_{str}(s\sigma)=
\begin{cases}
P_{str}(s)\gamma,&\mbox{if }\sigma\in \Gamma,\\
P_{str}(s)\sigma,&\mbox{if }\sigma\in \Sigma_o\backslash \Gamma,\\
P_{str}(s),&\mbox{otherwise}.
\end{cases}
\end{equation*}
where $s\in \Sigma^* \ \mbox{and} \ \sigma\in \Sigma$.$\hfill\Diamond$
\end{definition}


By definition, an intruder can observe the results after processing different types of events by the ideal anonymous processor (strong anonymous projection) as follows.
For an anonymous event, its observation is $\gamma$.
For an observable event that is not an anonymous event, its observation is itself, since there is no anonymous processor deployed on the sensor related to this event.
Obviously, the intruder observes nothing once an unobservable event occurs since the sensors associated with this event are not deployed.

By changing the natural projection $P$ in $\widehat{E}_{G}(P, \alpha|\alpha\beta)$ to the strong anonymous projection $P_{str}$, $\widehat{E}_{G}(P_{str}, \alpha|\alpha\beta)$ is obtained, representing the \emph{delayed state estimate under the strong anonymous projection} $P_{str}$ at the instant that $\alpha$ has just been exposed upon observing $\alpha\beta$.
Next, we use the notations introduced above to describe the scenario that an ideal anonymous processor is employed, i.e., an intruder cannot distinguish between any two anonymous events.
In this situation, the observations processed by the ideal anonymous processor are obtained by the intruder.
The intruder is expected not be able to determine for sure whether the system is in only one state $K$ observational steps before the current time instant.
As a result, the notion of $K$-step strong anonymity is formulated.

\begin{definition}\label{def:3.1}($K$-step strong anonymity)
Given a system $G=(X, \Sigma, \delta, X_0)$, a non-negative integer $K\in \mathbb{N}$, and the strong anonymous projection $P_{str}$ w.r.t. a set of anonymous events $\Gamma$ and a set of observable events $\Sigma_o$ with $\Gamma\subseteq \Sigma_o$, $G$ is said to be $K$-step strongly anonymous w.r.t. $\Sigma_o$, $\Gamma$, and $P_{str}$ if for all $st\in \mathcal{L}(G)$ with $P_{str}(s)=\alpha$, $P_{str}(st)=\alpha\beta$, and $|\beta|\leq K$, it holds
\begin{equation*}
|\widehat{E}_{G}(P_{str}, \alpha|\alpha\beta)|>1.\rightdiamond{}
\end{equation*}
\end{definition}
In plain words, $K$-step strong anonymity means that at any instant, an intruder cannot infer which state a system is exactly in at most $K$ observational steps before the current instant.
When a system satisfies $K$-step strong anonymity for any $K$, we say that the system satisfies infinite-step strong anonymity, as defined below.

\begin{definition}\label{def:3.2}(Infinite-step strong anonymity)
Given a system $G=(X, \Sigma, \delta, X_0)$ and the strong anonymous projection $P_{str}$ w.r.t. a set of anonymous events $\Gamma$ and a set of observable events $\Sigma_o$ with $\Gamma\subseteq \Sigma_o$, $G$ is said to be infinite-step strongly anonymous w.r.t. $\Sigma_o$, $\Gamma$, and $P_{str}$ if for all $st\in \mathcal{L}(G)$ with $P_{str}(s)=\alpha$ and $P_{str}(st)=\alpha\beta$, it holds
\begin{equation*}
|\widehat{E}_{G}(P_{str}, \alpha|\alpha\beta)|>1.\rightdiamond{}
\end{equation*}
\end{definition}

Intuitively, infinite-step strong anonymity requires that after observing $\alpha\beta$, the cardinality of the set of state estimates of system $G$ after $\alpha$ has just been observed inferred by an intruder are greater than $1$, so that the intruder cannot determine which state system $G$ is exactly in at any instant. 

\begin{example}\label{ex:3.1}
Consider the system $G$ shown in Fig. \ref{fig:3.1} and assume $K=1$.
To verify whether the system satisfies $1$-step strong anonymity, we need to check whether the cardinality of the delayed state estimate $\widehat{E}_{G}(P_{str}, \alpha|\alpha\beta)$ is greater than $1$ for all $\alpha\beta\in P_{str}(\mathcal{L}(G))$ with $|\beta|=0$ and $|\beta|=1$.
When $|\beta|=0$, taking the strings whose strong anonymous projections are $\gamma$ and $\gamma c$ for example respectively, one has $\widehat{E}_{G}(P_{str}, \gamma|\gamma)=\{2,3\}$ and $\widehat{E}_{G}(P_{str}, \gamma c|\gamma c)=\{4,5,6\}$.
It means that the number of states in the $0$ delayed state estimates is greater than 1 when an intruder observes $\gamma$ and $\gamma c$.
When $|\beta|=1$, one has $\widehat{E}_{G}(P_{str}, \epsilon|\gamma)=\{0,1\}$ and $\widehat{E}_{G}(P_{str}, \gamma|\gamma c)=\{2,3\}$.
The number of possible states the system $G$ could be in $1$ observational step ago is greater than $1$ when observing $\gamma$ and $\gamma c$.
Moreover, one finds that $|\widehat{E}_{G}(P_{str}, \alpha|\alpha\beta)|>1$ holds for all $\alpha\beta\in P_{str}(\mathcal{L}(G))$ with $|\beta|\leq 1$.
By Definition \ref{def:3.1}, the system $G$ is $1$-step strongly anonymous w.r.t. $\Sigma_o$, $\Gamma$, and $P_{str}$.

Further, we find that $|\widehat{E}_{G}(P_{str}, \alpha|\alpha\beta)|>1$ always holds regardless of the value of $K$, suggesting that the system $G$ also satisfies infinite-step strong anonymity.
It means that an intruder cannot accurately infer whether the system $G$ is in a singleton state at any instant.
$\hfill\square$
\end{example}

\begin{figure}[ht]
\centering
\includegraphics[width=1.9in]{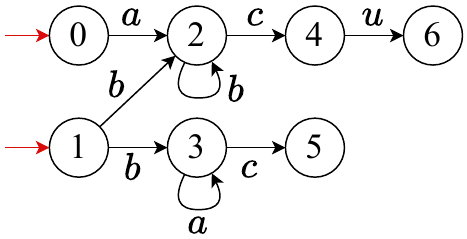}
\caption{System $G$ with $\Sigma_o=\{a, b, c\}$ and $\Gamma=\{a,b\}$.}
\label{fig:3.1}
\end{figure}

\subsection{K- and infinite-step weak anonymity}
Ideally, as stated in Definition \ref{def:3.1-2}, anonymous processors process different anonymous events performed by different agents with a fixed observational result $\gamma$.
However, anonymous processors may become less powerful (i.e., non-ideal) sometimes.
The non-ideal anonymous processor produces the same result for specific actions only performed by the same agent.
In other words, an intruder only knows when the same agent performs the specific actions without knowing which agent.

Some notations are introduced before formulating the non-ideal anonymous processor. 
Let $\Gamma_w=\{\gamma_1, \gamma_2, \dots, \gamma_{|\Gamma|}\}$ be a set of distinct observation labels of the anonymous events of all strings $s\in \mathcal{L}(G)$ under the weak anonymous projection.
In other words, set $\Gamma_w$ represents all possible labels for the first appearance of the anonymous events contained in string $s$.
We assume $\Gamma_w$ is disjoint with $\Sigma_o$, i.e., $\Gamma_w\cap \Sigma_o=\emptyset$ to distinguish anonymous events with the other observable events.
With a slight abuse of notation, write $\sigma\in s$ to denote that an event $\sigma\in \Sigma$ appears in string $s\in \Sigma^*$.
The set of anonymous events appearing in string $s$ is denoted by $E_{\Gamma}(s)=\Gamma\cap E(s)$, where $E(s)=\{\sigma\in \Sigma| \sigma\in s\}$. 
For any string $s=\mu\sigma v\in \Sigma^*$ with $\sigma\in \Gamma$, $\mu \in (\Sigma\backslash \{\sigma\})^*$, and $v\in \Sigma^*$, let $R(s, \sigma)=|E_{\Gamma}(\mu)|+1$ be the location index of the first occurrence of event $\sigma\in\Gamma$ in string $s$.
A non-ideal anonymous processor is formally characterized by the weak anonymous projection as follows.
\begin{definition}\label{def:3.2-2}
    A non-ideal anonymous processor, i.e., a processor whose processing result for an anonymous event in a string depends on the order of its first appearance in this string, and whose processing result for an observable event that is not anonymous event is still the event itself, is modeled as the \emph{weak anonymous projection} $P_w: \Sigma^* \rightarrow (\Sigma_o\cup \Gamma_w)^*$ (a dynamic observation function) formulated by
$P_w(\epsilon)=\epsilon$ and
\begin{equation*}
P_{w}(s\sigma)=
\begin{cases}
P_{w}(s)\gamma_{R(s\sigma, \sigma)},&\mbox{if }\sigma\in \Gamma,\\
P_{w}(s)\sigma,&\mbox{if }\sigma\in \Sigma_o\backslash \Gamma,\\
P_{w}(s),&\mbox{otherwise}.
\end{cases}
\end{equation*}
where $s\in \Sigma^* \mbox{~and~} \sigma\in \Sigma$.$\hfill\Diamond$
\end{definition}

In plain words, for $P_w$, the observation of an anonymous event depends on the order of its first appearance in a string, but for $P_{str}$, the observation of an anonymous event has nothing to do with the order of its first appearance.
As shown in Table \ref{tab:3.2}, to demonstrate the difference between $P_{str}$ and $P_w$ more intuitively, we list string $s\in \mathcal{L}(G)$, $P_{str}(s)$, and $P_w(s)$ in the first, second, and third columns, respectively, where system $G$ is shown in Fig. \ref{fig:3.1}.
From Table \ref{tab:3.2}, one sees that the strong anonymous projection of event $a/b$ is always fixed $\gamma$ due to its static property.
However, the weak anonymous projection has a dynamic characterization.
For example, the weak anonymous projection of the anonymous event $a$ in string $ab$ is $\gamma_1$, while in $ba$ is $\gamma_2$.
The reason is that the order in which event $a$ first appears in string $ab$ (resp., $ba$) is $1$ (resp., $2$) if we only focus on the order in which anonymous events first appear in string $ab$ (resp., $ba$).
Additionally, we finds that the strong anonymous projection of strings $ab$ and $bb$ (resp., strings $abc$ and $bbc$) are the same but their weak anonymous projections are different.
The technical details of calculating the weak anonymous projection for a string $s\in \mathcal{L}(G)$ are presented in Example \ref{ex:3.2}.

\begin{table}[htp]
\renewcommand{\arraystretch}{1.1}
\caption{Illustrative example}
\label{tab:3.2}
\centering
\resizebox{6.0cm}{!}{
\begin{tabular}{ccc}
\toprule  
string $s\in \mathcal{L}(G)$&$P_{str}(s)$&$P_w(s)$\\
\midrule 
$a/b$ & $\gamma$ & $\gamma_1$ \\ 
$ab/ba$ & $\gamma\gamma$ & $\gamma_1\gamma_2$ \\
$bb$ & $\gamma\gamma$ & $\gamma_1\gamma_1$ \\
$abc$ & $\gamma\gamma c$ & $\gamma_1\gamma_2 c$ \\
$bbc$ & $\gamma\gamma c$ & $\gamma_1\gamma_1 c$ \\
$abbc$ & $\gamma\gamma\gamma c$ & $\gamma_1\gamma_2\gamma_2 c$ \\
$\cdots$ & $\cdots$ & $\cdots$ \\
\bottomrule 
\end{tabular}}
\end{table}

\begin{example}\label{ex:3.2}
Consider again the system $G$ in Fig. \ref{fig:3.1} with $\Gamma=\{a,b\}$.
For string $s=b$, one has $E(s)=\{b\}$, $E_{\Gamma}(\epsilon)=\Gamma\cap E(\epsilon)=\emptyset$, and $R(b,b)=|E_{\Gamma}(\epsilon)|+1=1$.
Then, one concludes $P_w(b)=\gamma_{R(b,b)}=\gamma_1$ due to $b\in \Gamma$.
For string $t=ba$, $E(t)=\{a, b\}$, $E_{\Gamma}(s)=E_{\Gamma}(b)=\{b\}$, and $R(ba,a) = |E_{\Gamma}(b)| + 1 = 2$ can be obtained.
The weak anonymous projection of string $ba$ is $P_w(ba) = P_w(b)\gamma_{R(ba,a)} = \gamma_1\gamma_2$.
Similarly, for string $s'=a$, $P_w(s')= \gamma_{R(a, a)} =\gamma_1$ is derived owing to $a\in \Gamma$ and $R(a, a) = |E_\Gamma(\epsilon)| + 1 = 1$.
The weak anonymous projection of string $t'=ab$ is $P_w(ab) = P_w(s')\gamma_{R(ab,b)} = \gamma_1\gamma_2$ due to $b\in \Gamma$ and $R(ab, b)=|E_\Gamma(a)| + 1 = 2$.
As shown above, we have $P_w(ab)=P_w(ba)=\gamma_1\gamma_2$, which shows that $P_w$ fundamentally differs from $P_{str}$, due to $P_{str}(ab)=P_{str}(ba)=\gamma\gamma$.
$\hfill\square$
\end{example}

By changing the natural projection $P$ in $\widehat{E}_{G}(P, \alpha|\alpha\beta)$ to the weak anonymous projection $P_w$, $\widehat{E}_{G}(P_w, \alpha|\alpha\beta)$ is obtained, representing the \emph{delayed state estimate under the weak anonymous projection} $P_w$.
Consider the scenario that the processing capability of the anonymous processor is non-ideal.
It means that an intruder can only partially distinguish between anonymous events to a limited extent weakly.
In this case, the intruder obtains observations processed by a non-ideal anonymous processor.
We still expect that the intruder's estimate of the system state over the past $K$ observational steps is not a singleton state.
With the previously defined notations in mind, $K$-step weak anonymity is formally defined for a system $G$ as follows.

\begin{definition}\label{def:3.3}($K$-step weak anonymity)
Given a system $G=(X, \Sigma, \delta, X_0)$, a non-negative integer $K\in \mathbb{N}$, and the weak anonymous projection $P_w$ w.r.t. a set of anonymous events $\Gamma$ and a set of observable events $\Sigma_o$ with $\Gamma\subseteq \Sigma_o$, $G$ is said to be $K$-step weakly anonymous w.r.t. $\Sigma_o$, $\Gamma$, and $P_w$ if for all $st\in \mathcal{L}(G)$ with $P_w(s)=\alpha$, $P_w(st)=\alpha\beta$, and $|\beta|\leq K$, it holds
\begin{equation*}
|\widehat{E}_{G}(P_w, \alpha|\alpha\beta)|>1.\rightdiamond{}
\end{equation*}
\end{definition}
By definition, $K$-step weak anonymity under weak anonymous projection $P_{\omega}$ indicates that an intruder is not able to infer with certainty which state a system is in at most $K$ observational steps before the current instant, even if the intruder has weak discrimination on anonymous events.
Similarly, infinite-step weak anonymity can be obtained by removing the restriction on the length of the observational string $\beta$ in $K$-step weak anonymity.

\begin{definition}\label{def:3.4}(Infinite-step weak anonymity)
Given a system $G=(X, \Sigma, \delta, X_0)$ and the weak anonymous projection $P_w$ w.r.t. a set of anonymous events $\Gamma$ and a set of observable events $\Sigma_o$ with $\Gamma\subseteq \Sigma_o$, $G$ is said to be infinite-step weakly anonymous w.r.t. $\Sigma_o$, $\Gamma$, and $P_w$ if for all $st\in \mathcal{L}(G)$ with $P_w(s)=\alpha$ and $P_w(st)=\alpha\beta$, it holds
\begin{equation*}
|\widehat{E}_G(P_w, \alpha|\alpha\beta)|>1.\rightdiamond{}
\end{equation*}
\end{definition}

Similar to Definition \ref{def:3.2}, the infinite-step weak anonymity also characterizes that the number of state estimates at any past instant inferred by an intruder are greater than $1$.
Compared with the strong anonymity projection $P_{str}$, the weak anonymity projection $P_w$ provides a more detailed description of anonymous events in a string, which may lead to the situation that strong anonymous projections of two strings are the same, but their weak anonymous projections are different.
This also makes the verification of $K$- and infinite-step weak anonymity more complex.
The following example illustrates the difference between strong anonymity and weak anonymity.

\begin{example}\label{ex:3.3}
Reconsider the system $G$ shown in Fig. \ref{fig:3.1}.
Suppose $K=0$.
For the purpose of verifying $0$-step weak anonymity, we need to check whether $|\widehat{E}_{G}(P_w, \alpha\beta|\alpha\beta)| > 1$ holds for all $\alpha\beta\in P_w(\mathcal{L}(G))$.
Take string $s = bb$ for example.
One has $P_w(s) = \gamma_1\gamma_1$ and $|\widehat{E}_{G}(P_w, \gamma_1\gamma_1|\gamma_1\gamma_1)| = |\{2\}| = 1$.
It means that an intruder is sure that the system is in state $2$ currently (i.e., $0$ step ago) when observation $\gamma_1\gamma_1$ is obtained.
Obviously, system $G$ is not $0$-step weakly anonymous w.r.t. $\Sigma_o$, $\Gamma$, and $P_w$.
Hence, $G$ does not satisfy $K$-step weak anonymity for any $K$, suggesting that $G$ is not infinite-step weakly anonymous. 
Compared that $G$ is $1$-step strongly anonymous w.r.t. $\Sigma_o$, $\Gamma$, and $P_{str}$, one sees the huge difference between strong anonymity and weak anonymity.
$\hfill\square$
\end{example}

Similar to the proposed four types of anonymity, state-based opacity is also an information flow property based on state estimates.
In general, when state-based opacity is concerned, the states of a system are typically divided into secret states and non-secret states.
Opacity requires that an intruder can never be certain that a system is in secret states. 
It permits the intruder to determine that the system is in a specific non-secret state.
However, the proposed anonymity captures a different scenario compared to opacity.
It requires that the state estimates inferred by an intruder always contain at least two states.
In other words, the essence of the proposed anonymity lies in preventing any form of state determination so as to the maintenance of state uncertainty.
Moreover, in this paper, we consider this security and privacy requirement for the past $K$- and infinite observational steps under both ideal and non-ideal anonymous processors, respectively.
In summary, opacity and the proposed anonymity are two types of information flow properties that are used to characterize different security and privacy requirements.
Specifically, opacity focuses on the ambiguity of the secret states of a system, while the proposed anonymity emphasizes the non-uniqueness of all states of a system.
Based on their respective characteristics, opacity applies to the scenarios requiring ``selective confidentiality", while the proposed anonymity applies to those requiring ``universal confidentiality" or ``identity protection".

\begin{remark}\label{re:3.2}
Here, we discuss the distinction and relationship among different anonymity notions.
When $\Gamma\neq \emptyset$, the proposed four notions of anonymity are different from those of \cite{Wu2014} and \cite{Basilio2021} since it depends on how $\Gamma$ is defined.
If $\Gamma=\emptyset$, i.e., the strong and weak anonymous projections reduce to the natural projection $P$, then the notions of $K$-step (resp., infinite-step) strong anonymity and weak anonymity reduce to $K$-step (resp., infinite-step) anonymity defined in \cite{Basilio2021}.
In this case, to verify whether a system satisfies $K$-step anonymity (resp., infinite-step) anonymity, we only need to check whether $|\widehat{E}_{G}(P, \alpha|\alpha\beta)|>1$ holds or not, where $|\beta|\leq K$ (resp., $|\beta|$ has no restriction).
For instance, consider the system $G$ shown in Fig. \ref{fig:3.1} with $\Gamma=\emptyset$.
After observing $abc$, an intruder certainly infers that the system is in state $2$ after $ab$ occurs, i.e., $|\widehat{E}_{G}(P, ab|abc)|=|\{2\}|=1$.
Thus, system $G$ does not satisfy $1$-step anonymity, suggesting that $G$ also does not satisfy infinite-step anonymity.

Further, if $\Gamma=\emptyset$ and $K=0$, both $K$-step strong anonymity and weak anonymity in Definitions \ref{def:3.1} and \ref{def:3.3} reduce to current-state anonymity defined in \cite{Wu2014} and \cite{Basilio2021}.
The requirement of current-state anonymity is that the set of possible current states of the system is never a singleton set.
For example, reconsider the system $G$ shown in Fig. \ref{fig:3.1} with $\Gamma=\emptyset$.
An intruder determines that the system is in state $3$ after observing $ba$, i.e., $|\widehat{E}_{G}(P, ba|ba)|=|\{3\}|=1$, which violates current-state anonymity.

Overall, the proposed four types of anonymity are more generalized versions of those of existing anonymity.
In other words, the existing current-state anonymity, $K$-step anonymity, and infinite-step anonymity can be viewed as their special cases.
In addition, by definitions, the relationship between $K$-step (resp., infinite-step) strong anonymity and weak anonymity is that $G$ is $K$-step (resp., infinite-step) strongly anonymous if it is $K$-step (resp., infinite-step) weakly anonymous, but not vice versa.
$\hfill\square$
\end{remark}

\begin{proposition}\label{pro:3.1}
Given a system $G=(X, \Sigma,\delta, X_0)$, a non-negative integer $K\in \mathbb{N}$, a set of observable events $\Sigma_o$, a set of anonymous events $\Gamma$ with $\Gamma\subseteq \Sigma_o$, and the strong and weak anonymous projection $P_{str}$ and $P_w$, the following statements hold:\\
1) system $G$ is $\widetilde{K}$-step strongly (resp., weakly) anonymous w.r.t. $\Sigma_o$, $\Gamma$, and $P_{str}$ (resp., $P_w$) if $G$ is $K$-step strongly (resp., weakly) anonymous w.r.t. $\Sigma_o$, $\Gamma$, and $P_{str}$ (resp., $P_w$), and $0 \leq \widetilde{K} < K$;\\
2) it is $K$-step strongly (resp., weakly) anonymous w.r.t. $\Sigma_o$, $\Gamma$, and $P_{str}$ (resp., $P_w$) for any $K\in \mathbb{N}$ if $G$ is infinite-step strongly (resp., weakly) anonymous w.r.t. $\Sigma_o$, $\Gamma$, and $P_{str}$ (resp., $P_w$).
\end{proposition}
\begin{proof}
By Definitions \ref{def:3.1}, \ref{def:3.2}, \ref{def:3.3}, and \ref{def:3.4}, statements 1) and 2) are trivially true.
\end{proof}

\section{Verification of anonymity using concurrent compositions}\label{sec:4}
\subsection{Verification of K- and infinite-step strong anonymity}
In this subsection, we focus on the verification of $K$- and infinite-step strong anonymity.
According to the definition of strong anonymous projection, the observer of system $G$ is constructed as a DFA, denoted by
\begin{equation}\label{eq:3}
Obs(G) = (X_{obs}, \Sigma_{obs}, \delta_{obs}, X_{0,obs}),
\end{equation}
where
\begin{itemize}
\item{$X_{obs}= 2^{X}$ is a finite set of states\footnote{For the ease of showing anonymity verification algorithms, the set $X_{obs}$ includes all states that are reachable and unreachable (from $X_{0,obs}$). 
When depicting an automaton, we usually only draw its reachable part or its part related to verification of the corresponding definition of anonymity. 
In the sequel, we will define $Obs_m(G)$ from the reachable states in $X_{obs}$.},}
\item{$\Sigma_{obs}=(\Sigma_o\backslash \Gamma) \cup \{\gamma\}$ is a finite set of observable events,}
\item{$\delta_{obs}: X_{obs} \times \Sigma_{obs} \rightarrow X_{obs}$ is the transition function defined by: for any $x_{obs} \in X_{obs}$ and $\sigma_{obs} \in \Sigma_{obs}$, one has $\delta_{obs}(x_{obs}, \sigma_{obs})=UR(\{x' \in X|\exists x\in x_{obs}, \exists \sigma\in \Gamma: P_{str}(\sigma)=\gamma \land x' \in \delta(x, \sigma)\})$ if $\sigma_{obs}=\gamma$, $\delta_{obs}(x_{obs}, \sigma_{obs})=UR(\{x' \in X|\exists x\in x_{obs}: x' \in \delta(x, \sigma_{obs})\})$ otherwise, and}
\item{$X_{0,obs}= UR(X_0)$ is the set of initial states.}
\end{itemize}

In observer $Obs(G)$, the adjoint state set of state $x$ in terms of state $x_{obs}$ is defined by 
\begin{equation*}
As(x, x_{obs})=
\begin{cases}
UR(x_{obs}\backslash\{x\}),&\mbox{if~} x\in x_{obs},\\
\mbox{undefined},&\mbox{otherwise}.
\end{cases}
\end{equation*}
We can tabulate all adjoint states of each state in all $x_{obs}\in X_{obs}$.
Note that there may be more than one adjoint state set for state $x$ since $x$ may be contained in different states of $X_{obs}$.
Then, we construct a new observer of system $G$ called a \emph{modified observer}, defined by
\begin{equation}\label{eq:4}
Obs_m(G) = (X_{m}, \Sigma_{m}, \delta_{m}, X_{0,m}),
\end{equation}
where
\begin{itemize}
\item{$X_{m}= 2^{X}$ is a finite set of states,}
\item{$\Sigma_{m}=\Sigma_{obs}$ is a finite set of observable events,}
\item{$\delta_{m}: X_{m} \times \Sigma_m \rightarrow X_{m}$ is the transition function defined by: for any $x_{m} \in X_m$ and $\sigma_{m} \in \Sigma_{m}$, $\delta_{m}(x_{m}, \sigma_{m}) = UR(\{x' \in X|\exists x\in x_{m}, \exists \sigma\in \Gamma: P_{str}(\sigma)=\gamma \land x' \in \delta(x, \sigma)\})$ if $\sigma_m = \gamma$, $\delta_{m}(x_{m}, \sigma_{m}) = UR(\{x' \in X|\exists x\in x_{m}: x' \in \delta(x, \sigma_{m})\})$ otherwise, and}
\item{$X_{0,m}=\bigcup_{\mycom{x\in x_{obs}}{\emptyset\neq x_{obs} \in \Bar X_{obs}}} \{As(x, x_{obs})\}$ is the set of initial 

\vspace{0.3em}
states, where $\Bar X_{obs}$ is the set of reachable states in $Obs(G)$ (note that $\emptyset \in \Bar X_{obs}$ may hold.).}
\end{itemize}

Clearly, the only difference between $Obs(G)$ and $Obs_m(G)$ is the set of initial states.
Notation $Obs_m(G)$ characterizes the evolution of all non-empty adjoint state sets during all observational steps, which is critical for the subsequent decision on the existence of adjoint states within $K$ observational steps.

A variant of system $G$, denoted by $G_m$, is realized by letting all the states in $G$ be initial states.
Next, we introduce the main tool--concurrent composition.

\begin{definition}(Concurrent composition)\label{de:200}
Consider two automata $G^i=(X_i, \Sigma_i, \delta_i, X_{0,i})$, $i=1, 2$, where $G^1$ and $G^2$ are NFAs.
Let $\ell: \Sigma_{1} \rightarrow \Sigma_{2}\cup \{\epsilon\}$ be a projection mapping, where all events in $\Sigma_{2}$ are observable.
The concurrent composition of $G^1$ and $G^2$ is
\begin{equation}\label{eq:200}
Cc(G^1, G^2)=(X^c, \Sigma^c, \delta^c, X^{0,c}),
\end{equation}
where\footnote{``$Cc$'' is the abbreviation of ``concurrent composition''.}
\begin{itemize}
    \item $X^c= X_1 \times X_2$,
    \item $\Sigma^c= \Sigma^{o,c} \cup \Sigma^{uo,c}$, where $\Sigma^{o,c}=\{(\overline{\sigma}, \overline{\sigma}')| \overline{\sigma}\in \Sigma_{1o} \land \overline{\sigma}'\in \Sigma_{2}\land \ell(\overline{\sigma})=\overline{\sigma}'\}$, $\Sigma^{uo,c}=\{(\overline{\sigma}, \epsilon)| \overline{\sigma}\in \Sigma_{1uo}\}$, and $\Sigma_{1o}$ (resp., $\Sigma_{1uo}$) is the set of observable (resp., unobservable) events of $\Sigma_1$,
    \item for all $(x_1, x_2)\in X^c$, $(\overline{\sigma}, \overline{\sigma}') \in \Sigma^{o,c}$, and $(\overline{\sigma}, \epsilon)\in \Sigma^{uo,c}$,
        \begin{itemize}
        \item [1)] $(x'_1, x'_2)\in \delta^c((x_1, x_2), (\overline{\sigma}, \overline{\sigma}'))$ if and only if $x'_1\in \delta_1(x_1, \overline{\sigma})$ and $x'_2=\delta_2(x_2, \overline{\sigma}')$,
        \item [2)] $(x'_1, x'_2)\in \delta^c((x_1, x_2), (\overline{\sigma}, \epsilon))$ if and only if $x'_1\in \delta_1(x_1, \overline{\sigma})$ and $x'_2=x_2$,
        \end{itemize}
    \item $X^{0,c}\subseteq X_{0,1}\times X_{0,2}$ is the set of initial states and will be specified when defining (\ref{eq:5}) and (\ref{eq:9}).$\hfill\Diamond$
\end{itemize}
\end{definition}

To verify $K$- and infinite-step strong anonymity, we construct the concurrent composition of $G_m$ and $Obs_m(G)$ by considering projection mapping $\ell$ as the strong anonymous projection $P_{str}$.
The constructed concurrent composition is accordingly denoted by
\begin{equation}\label{eq:5}
Cc(G_m, Obs_m(G)) = (X_c, \Sigma_c, \delta_c, X_{0,c}),
\end{equation}
where $G_m$ is the variant automaton for $G$ by letting all the states in $G$ be initial states, $Obs_m(G)$ is the modified observer as defined in (\ref{eq:4}) (Note that $Obs_m(G)$ is not an observer of $G_m$.),
and the set of initial states $X_{0,c}=\bigcup_{\mycom{x\in x_{obs}}{\emptyset\ne x_{obs} \in \Bar X_{obs}}} \{(x, As(x, x_{obs}))\}$, where $\Bar X_{obs}$ is the set of reachable states in $Obs(G)$.
A state $x_c$ in $Cc(G_m, Obs_m(G))$ takes the form of $x_c=(x, q)$, where the left component $x$ stands for the actual state of the system $G$, while the right component $q$ represents the corresponding adjoint state set of the left component $x$.
For any string $s_c\in \mathcal{L}(Cc(G_m, Obs_m(G)))$,
the left and right components of $s_c$ are denoted as $L(s_c)$ and $R(s_c)$, respectively.
By the construction of $Cc(G_m, Obs_m(G))$, we find $P_{str}(L(s_c)) = R(s_c)$.
Next, we illustrate the construction of $Cc(G_m, Obs_m(G))$ for system $G$.

\begin{example}\label{ex:3.4}
Consider again the system $G$ in Fig. \ref{fig:3.1}.
The observer $Obs(G)$ is shown in Fig. \ref{fig:4.1}.
From the constructed $Obs(G)$, we obtain the states and their corresponding adjoint state sets as shown in Table \ref{tab:4.1}.
Further, the modified observer $Obs_m(G)$ and the variant automaton $G_m$ are shown in Figs. \ref{fig:4.2} and \ref{fig:4.3}, respectively.
The concurrent composition $Cc(G_m, Obs_m(G))$ is constructed as shown in Fig. \ref{fig:4.4}.
$\hfill\square$
\end{example}

\begin{figure}[htp]
\centering
\includegraphics[width=2.1in]{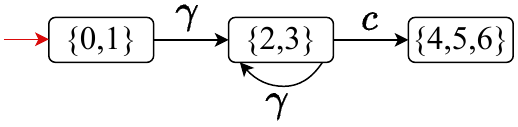}
\caption{Part of the observer $Obs(G)$ for $G$.}
\label{fig:4.1}
\end{figure}

\begin{table}[htp]
\renewcommand{\arraystretch}{1.2}
\caption{States and the corresponding adjoint state sets}
\label{tab:4.1}
\centering
\resizebox{5.2cm}{!}{
\begin{tabular}{ccc}
\toprule  
$x_{obs}$&state&adjoint state set\\
\midrule 
\multirow{2}{*}{\{0,1\}}&0&\{1\} \\ 
&1&\{0\} \\
\midrule
\multirow{2}{*}{\{2,3\}}&2&\{3\} \\ 
&3&\{2\} \\
\midrule
\multirow{3}{*}{\{4,5,6\}}&4&\{5,6\} \\ 
&5&\{4,6\} \\
&6&\{4,5,6\} \\
\bottomrule 
\end{tabular}}
\end{table}

\begin{figure}[htp]
\centering
\includegraphics[width=2.0in]{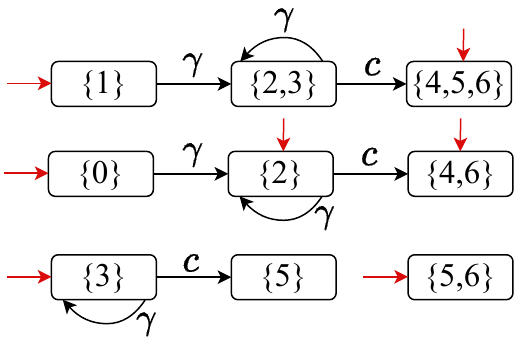}
\caption{Part of the modified observer $Obs_m(G)$.}
\label{fig:4.2}
\end{figure}

\begin{figure}[htp]
\centering
\includegraphics[width=1.7in]{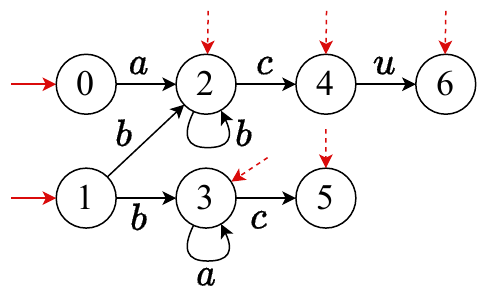}
\caption{The variant automaton $G_m$ for $G$.}
\label{fig:4.3}
\end{figure}

\begin{figure}[!ht]
\centering
\includegraphics[width=3.2in]{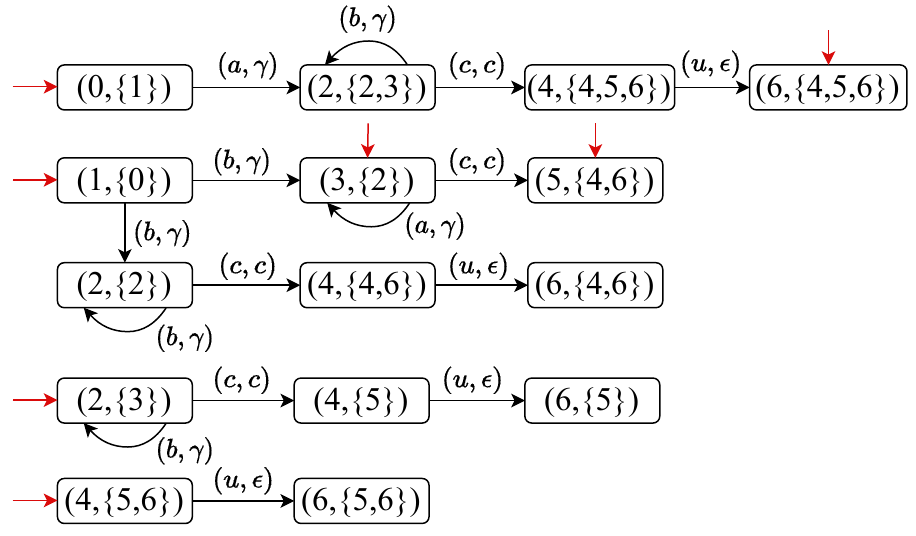}
\caption{Part of the concurrent composition $Cc(G_m,Obs_m(G))$.}
\label{fig:4.4}
\end{figure}

Now we present the main results of verifying $K$- and infinite-step strong anonymity using the proposed concurrent composition $Cc(G_m, Obs_m(G))$.

\begin{theorem}\label{th:4.1}
Given a system $G=(X, \Sigma, \delta, X_0)$, let $Cc(G_m, Obs_m(G))  = (X_c, \Sigma_c, \delta_c, X_{0,c})$ be the corresponding concurrent composition.
$G$ is infinite-step (resp., $K$-step) strongly anonymous w.r.t. $\Sigma_o$, $\Gamma$, and $P_{str}$ if and only if there exists no state of the form $(\cdot, \emptyset)$ in $Cc(G_m, Obs_m(G))$ that is reachable from $X_{0,c}$ (resp., within $K$ observational steps).
\end{theorem}
\begin{proof}
\rm{First, we show the case of $K$-step strong anonymity in Theorem \ref{th:4.1}.\\
(if) By contrapositive, suppose that system $G$ is not $K$-step strongly anonymous w.r.t. $\Sigma_o$, $\Gamma$, and $P_{str}$.
By Definition \ref{def:3.1}, there exists a string $st\in \mathcal{L}(G)$ such that $|\widehat{E}_{G}(P_{str}, \alpha|\alpha\beta)| > 1$ does not hold, where $P_{str}(s) = \alpha$, $P_{str}(st) = \alpha\beta$, and $|\beta| \leq K$.
Then $|\widehat{E}_{G}(P_{str}, \alpha|\alpha\beta)| = 1$.  Denote $\widehat{E}_{G}(P_{str}, \alpha|\alpha\beta) = \{x_c\}$.
By the construction of $Obs(G)$, $x_c\in x_{obs}$ where $x_{obs}=\delta_{obs}(X_{0,obs}, \alpha)$.
By the construction of $Obs_m(G)$, $\delta_m(As(x_c, x_{obs}), \beta)=\emptyset$.
By the construction of $Cc(G_m, Obs_m(G))$, there exists a string $s_c\in \mathcal{L}(Cc(G_m, Obs_m(G)))$ with $R(s_c)=\beta$ such that $(x_{c,k} ,\emptyset)\in \delta_c((x_c, As(x_c, x_{obs})), s_c)$, where $x_{c,k}\in \delta(x_c, L(s_c))$.
One concludes that $(x_{c, k}, \emptyset)$ is reachable from $X_{0,c}$ within $K$ observational steps due to $|P_{str}(L(s_c))| = |R(s_c)| \leq K$.\\
(only if) Still by contrapositive, suppose that there exists a state of the form $(\cdot, \emptyset)$ in $Cc(G_m, Obs_m(G))$ that is reachable from $X_{0,c}$ within $K$ observational steps.
By the construction of $Cc(G_m,Obs_m(G))$, there exists a state $(x_{c,k}, \emptyset)\in X_c$ reached from some initial state $(x_c, As(x_c, x_{obs}))\in X_{0,c}$ and a string $s_c\in \mathcal{L}(Cc(G_m, Obs_m(G)))$ such that $|R(s_c)|\leq K$, $x_c\in x_{obs}$, $x_{c,k}\in \delta(x_c, L(s_c))$, and $(x_{c,k}, \emptyset)\in \delta_c((x_c, As(x_c, x_{obs})), s_c)$.
By the construction of $Obs_m(G)$, $\delta_m(As(x_c, x_{obs}), R(s_c))=\emptyset$.
Further, $\delta(x'_c, t)=\emptyset$ holds for all $x'_c\in As(x_c, x_{obs})$ and all $t\in \Sigma^*$ with $P_{str}(t)=R(s_c)$.
Then, by the construction of $Obs(G)$, one concludes that there exists an initial state $x_0\in X_0$ and a string $s\in \mathcal{L}(G, x_0)$ such that $x_c \in \delta(x_0, s)$ and $\delta_{obs}(X_{0,obs}, P_{str}(s))=x_{obs}$ owing to $x_c\in x_{obs}$.
Let $P_{str}(s)=\alpha$ and $P_{str}(t)=\beta$.
Accordingly, $x_c$ is the unique state in $\widehat{E}_{G}(P_{str}, \alpha|\alpha\beta)$, i.e., $|\widehat{E}_{G}(P_{str}, \alpha|\alpha\beta)|=1$.
By Definition \ref{def:3.1}, system $G$ does not satisfy $K$-step strong anonymity w.r.t. $\Sigma_o$, $\Gamma$, and $P_{str}$.

The proof for the case of infinite-step strong anonymity is very similar to the mentioned case, we omit the details here.}
\end{proof}

By Theorem \ref{th:4.1}, we summarize the verification procedures for $K$- and infinite-step strong anonymity as shown in Algorithm \ref{al:1}.
In Algorithm \ref{al:1}, the time costs of computing $Obs(G)$, $Obs_m(G)$, $G_m$, and $Cc(G_m, Obs_m(G))$ are $\mathcal{O}(|\Sigma_o||\Sigma_{uo}||X|^2 2^{|X|})$, $\mathcal{O}(|\Sigma_o||\Sigma_{uo}||X|^2 2^{|X|})$, $\mathcal{O}(|\Sigma||X|^2)$, and $\mathcal{O}(|\Sigma||X|^2 2^{|X|})$, respectively.
Hence, the overall time complexity for verifying $K$- and infinite-step strong anonymity using Algorithm \ref{al:1} is $\mathcal{O}( (|\Sigma_o||\Sigma_{uo}|+|\Sigma|) |X|^2 2^{|X|})$.
Note that the two-way observer method proposed in \cite{Yin2017} can also verify $K$- and infinite-step strong anonymity by making slight modifications, but its complexity is higher than ours.

\begin{algorithm}[!htbp]
  \DontPrintSemicolon
  \caption{Verification of Infinite-Step (resp., $K$-Step) Strong Anonymity}\label{al:1}
  \KwIn{System $G=(X, \Sigma, \delta, X_0), \Sigma_o, \Sigma_{uo}, \Gamma \subseteq \Sigma_o$, and $K\in \mathbb{N}$.}
  \KwOut{``YES" if $G$ satisfies infinite-step (resp., $K$-step) strong anonymity w.r.t. $\Sigma_o$, $\Gamma$, and $P_{str}$, ``NO" otherwise.}
  Construct the observer $Obs(G)$ for $G$;\;
  \eIf{there exists a reachable state $x_{obs}\in X_{obs}$ such that $|x_{obs}|=1$,}
  {\textbf{return} ``No'';\;
	\textbf{stop};\;
  }{
	Construct modified observer $Obs_m(G)$ for $G$;\;
	Construct $G_m$ as a variant of $G$;\;
	Construct $Cc(G_m, Obs_m(G))$ as the concurrent composition of $G_m$ and $Obs_m(G)$;\;
	Use the ``Breadth-First Search Algorithm” in \cite{Cormen2022} to decide whether there exists a state of the form $(\cdot,\emptyset)$ in $Cc(G_m, Obs_m(G))$ that is reachable from $X_{0,c}$ (resp., within $K$ observational steps);\;
	\eIf{there exists such a state $(\cdot,\emptyset)$ in $Cc(G_m, Obs_m(G))$ that is reachable from $X_{0,c}$ (resp., within $K$ observational steps),}
	{\textbf{return} ``No'';\;
	\textbf{stop};\;
	}
	{\textbf{return} ``Yes'';\;
	\textbf{stop}.
	}
  }
\end{algorithm}


\begin{example}
Reconsider the system $G$ in Fig. \ref{fig:3.1}. 
The corresponding concurrent composition $Cc(G_m, Obs_m(G))$ is shown in Fig. \ref{fig:4.4} in which there exists no reachable state of the form $(\cdot, \emptyset)$, concluding that $G$ satisfies infinite-step strong anonymity.
It also satisfies $K$-step strong anonymity for any value of $K$.
The verification results are consistent with the conclusions in Example \ref{ex:3.1}.
$\hfill\square$
\end{example}


\subsection{Verification of K- and infinite-step weak anonymity}
As the weak anonymous projection $P_w$ is a class of dynamic observation functions, the verification for $K$- and infinite-step anonymity under the weak anonymous projection $P_w$ is more complicated than their verification under the strong anonymous projection $P_{str}$.
In other words, the previous $Cc(G_m, Obs_m(G))$ cannot be used to verify $K$- and infinite-step weak anonymity.
To this end, we need an augmented automaton to express historical information about the occurrence of anonymous events along with their corresponding location indices.
Some symbols are introduced before constructing the augmented automaton of the system $G$.
For the anonymous event set $\Gamma$, we define its augmented set as $\widehat{\Gamma}=\bigcup_{\sigma\in \Gamma} \{\sigma_i| i=1, \dots, |\Gamma|\}$.
For example, one has $\widehat{\Gamma}=\{a_1,a_2,b_1,b_2\}$ if $\Gamma=\{a,b\}$.
For system $G$, its augmented system is defined by
\begin{equation}\label{eq:6}
\widehat{G} = (\widehat{X}, \widehat{\Sigma}, \widehat{\delta}, \widehat{X}_0),
\end{equation}
where
\begin{itemize}
\item{$\widehat{X}= (X \times \{\epsilon\}) \cup \bigcup_{i=1}^{|\Gamma|} (X \times \{\{(\sigma^1)_1, \dots (\sigma^i)_i\}| \sigma^1, \dots,\\ \sigma^i\in \Gamma \mbox{ are distinct}\})$ is a finite set of states,} 
\item{$\widehat{\Sigma}=\widehat{\Gamma}\cup (\Sigma_o\backslash \Gamma)\cup \Sigma_{uo}$ is a finite set of events, where $\widehat{\Sigma}=\widehat{\Sigma}_{o}\dot\cup\widehat{\Sigma}_{uo}$, $\widehat{\Sigma}_{o}=\widehat{\Gamma}\cup (\Sigma_o\backslash \Gamma)$, and $\widehat{\Sigma}_{uo}=\Sigma_{uo}$,}
\item{$\widehat{\delta}: \widehat{X} \times \widehat{\Sigma} \rightarrow 2^{\widehat{X}}$ is defined by: \\
(1) for any state $(x, p)\in \widehat{X}$, $e\in (\Sigma_o\backslash \Gamma)\cup \Sigma_{uo}$, let $\widehat{\delta}((x, p), e)=\{(x', p)| x'\in \delta(x, e)\}$, \\
(2) for all $1\leq i \leq |\Gamma|$, $1\leq j \leq i$,
for all $(x, \{(\sigma^1)_1, \dots, (\sigma^i)_i\})\in \widehat{X}$, $(\sigma^j)_j\in \widehat{\Gamma}$, 
let $\widehat{\delta}((x, \{(\sigma^1)_1, \dots,(\sigma^i)_i\}), (\sigma^j)_j)=\{(x', \{(\sigma^1)_1, \dots, (\sigma^i)_i\})| x'\in \delta(x, \sigma^j)\}$,\\
(3) for all $0\leq i \leq |\Gamma|$, 
for all $(x, \{(\sigma^1)_1, \dots, (\sigma^i)_i\})\in \widehat{X}$, $(\sigma^{i+1})_{i+1}\in \widehat{\Gamma}$ with $\sigma^1, \dots, \sigma^{i+1}$ being distinct,
let $\widehat{\delta}((x, \{(\sigma^1)_1, \dots, (\sigma^{i})_{i}\}), (\sigma^{i+1})_{i+1}) = \{(x', \{(\sigma^1)_1, \dots, (\sigma^i)_i, (\sigma^{i+1})_{i+1}\})| x'\in \delta(x, \sigma^{i+1})\}$,\\
(4) $\widehat{\delta}$ is not defined at any other pair in $\widehat{X} \times \widehat{\Sigma}$, and
}
\item{$\widehat{X}_0=X_0 \times \{\epsilon\}$ is a finite set of initial states.}
\end{itemize}

We define a projection mapping $P^m: \widehat{\Gamma}\rightarrow \Gamma_w$ as
$P^m(\sigma_i)=\gamma_i$ for all $\sigma_i\in \widehat{\Gamma}$ with $1\leq i\leq |\Gamma|$ to replace an event $\sigma_i$ in $\widehat{\Gamma}$ by $\gamma_i$, where $\gamma_i\in \Gamma_w$ and $\Gamma_w=\{\gamma_1, \gamma_2, \dots, \gamma_{|\Gamma|}\}$.
To describe the language generated by augmented automaton $\widehat{G}$ under the weak anonymous projection $P_w$, define a projection mapping $P^n: \widehat{\Sigma}^* \rightarrow (\Gamma_w\cup(\Sigma_o\backslash\Gamma))^*$ as $P^n(\epsilon)=\epsilon$ and
\begin{equation*}
P^n(s\sigma)=
\begin{cases}
P^n(s)P^m(\sigma),&\mbox{if }\sigma\in \widehat{\Gamma},\\
P^n(s)\sigma,&\mbox{if }\sigma\in (\Sigma_o\backslash \Gamma),\\
P^n(s),&\mbox{otherwise}.
\end{cases}
\end{equation*}
where $s\in \widehat{\Sigma}^* \ \mbox{and} \ \sigma\in \widehat{\Sigma}$.
Trivially, $P^n(\mathcal{L}(\widehat{G}))=P_w(\mathcal{L}(G))$ holds.
In augmented automaton $\widehat{G}$, the unobservable reach of a state set $Q\subseteq \widehat{X}$ is defined as
\begin{equation*}
UR_w(Q)=\{q'\in \widehat{X}|\exists q\in Q, \exists s\in \Sigma_{uo}^*: q'\in \widehat{\delta}(q, s)\}.
\end{equation*}
Then, the observer for $\widehat{G}$ is constructed as a DFA, denoted by
\begin{equation}\label{eq:7}
\widehat{Obs}(\widehat{G}) = (\widehat{X}_{obs}, \widehat{\Sigma}_{obs}, \widehat{\delta}_{obs}, \widehat{X}_{0,obs}),
\end{equation}
where
\begin{itemize}
\item{$\widehat{X}_{obs}= 2^{\widehat{X}}$ is a finite set of states,}
\item{$\widehat{\Sigma}_{obs}=\Gamma_w\cup (\Sigma_{o}\backslash \Gamma)$ is a finite set of observable events,} 
\item{$\widehat{\delta}_{obs}: \widehat{X}_{obs} \times \widehat{\Sigma}_{obs} \rightarrow \widehat{X}_{obs}$ is the transition function defined by: for any $\widehat{x}_{obs} \in \widehat{X}_{obs}$ and $\widehat{\sigma} \in \widehat{\Sigma}_{obs}$, one has $\widehat{\delta}_{obs}(\widehat{x}_{obs}, \widehat{\sigma})=UR_w(\{\widehat{x}' \in \widehat{X}|\exists \widehat{x} \in \widehat{x}_{obs}: \widehat{x}' \in \widehat{\delta}(\widehat{x}, \widehat{\sigma})\})$ if $\widehat{\sigma}\in (\Sigma_{o}\backslash \Gamma)$; $\widehat{\delta}_{obs}(\widehat{x}_{obs}, \widehat{\sigma})=UR_w(\{\widehat{x}' \in \widehat{X}|\exists \widehat{x}\in \widehat{x}_{obs}, \exists \widehat{\sigma}' \in \widehat{\Gamma}: \widehat{x}' \in \widehat{\delta}(\widehat{x}, \widehat{\sigma}') \land P^m(\widehat{\sigma}') = \widehat{\sigma}\})$ otherwise, and}
\item{$\widehat{X}_{0,obs}= UR_w(\widehat{X}_0)$ is the set of initial states.}
\end{itemize}

In order to describe the state $x\in X$ involved in $\widehat{x}_{obs}\in \widehat{X}_{obs}$, let $m(\widehat{x}_{obs})=\{x\in X| (x, \cdot)\in \widehat{x}_{obs}\}$.
For any state $\widehat{x}=(x,\cdot)\in \widehat{x}_{obs}$, we use its left part to denote its corresponding original state $x$, denoted by $l(\widehat{x})=x$.
Write $n(\widehat{x}, \widehat{x}_{obs})=\{(x,\cdot)\in \widehat{x}_{obs}|x= l(\widehat{x})\}$ as the set of states in $\widehat{x}_{obs}$ that corresponding to the same original state $x$.
The adjoint state set of state $\widehat{x}$ in terms of state $\widehat{x}_{obs}$ in observer $\widehat{Obs}(\widehat{G})$ is defined as
\begin{equation*}
AS_w(\widehat{x}, \widehat{x}_{obs})=
\begin{cases}
UR_w(\widehat{x}_{obs}\backslash n(\widehat{x}, \widehat{x}_{obs})),&\mbox{if~} \widehat{x}\in \widehat{x}_{obs}\land\\
&|m(\widehat{x}_{obs})|>1,\\
\emptyset, &\mbox{if~} \widehat{x}\in \widehat{x}_{obs}\land \\
&|m(\widehat{x}_{obs})|=1,\\
\mbox{undefined},&\mbox{otherwise}.
\end{cases}
\end{equation*}
Also, a variant of $\widehat{Obs}(\widehat{G})$ is derived by changing the set of initial states, defined by
\begin{equation}\label{eq:8}
Obs_w(\widehat{G}) = (X_{w}, \Sigma_{w}, \delta_{w}, X_{0,w}),
\end{equation}
where
\begin{itemize}
\item{$X_{w}= 2^{\widehat{X}}$ is a finite set of states,}
\item{$\Sigma_{w}=\widehat{\Sigma}_{obs}$ is a finite set of observable events,}
\item{$\delta_{w}: X_{w} \times \Sigma_w \rightarrow X_{w}$ is the transition function defined by: for any $x_{w} \in X_w$ and $\sigma_{w} \in \Sigma_{w}$, one has $\delta_{w}(x_{w}, \sigma_{w}) = UR_w(\{\widehat{x}' \in \widehat{X}|\exists \widehat{x} \in x_{w}: \widehat{x}' \in \widehat{\delta}(\widehat{x}, \sigma_w)\})$ if $\sigma_w \in (\Sigma_{o}\backslash \Gamma)$; $\delta_{w}(x_{w}, \sigma_{w})=UR_w(\{\widehat{x}' \in \widehat{X}|\exists \widehat{x}\in x_w, \exists \widehat{\sigma}' \in \widehat{\Gamma}: \widehat{x}' \in \widehat{\delta}(\widehat{x}, \widehat{\sigma}') \land P^m(\widehat{\sigma}') = \sigma_w\})$ otherwise, and}
\item{$X_{0,w}=\bigcup_{\mycom{\hat{x}\in \hat{x}_{obs}} {\emptyset\neq \hat{x}_{obs} \in \widetilde{X}_{obs}}} \{AS_w(\widehat{x}, \widehat{x}_{obs})\}$ is the initial state

\vspace{0.3em}
sets, where $\widetilde{X}_{obs}$ is the set of reachable states in $\widehat{Obs}(\widehat{G})$.}
\end{itemize}

Let all the states in $\widehat{G}$ be initial states, which enables us to get a new NFA $\widehat{G}_w$.
Now, the concurrent composition of $\widehat{G}_w$ and $Obs_w(\widehat{G})$ is constructed by considering projection mapping $\ell$ as the projection mapping $P^n$.
The constructed concurrent composition is represented by
\begin{equation}\label{eq:9}
Cc(\widehat{G}_w, Obs_w(\widehat{G})) = (X_c^w, \Sigma_c^w, \delta_c^w, X_{0,c}^w),
\end{equation}
where $\widehat{G}_w$ is the variant automaton for $\widehat{G}$ by letting all states in $\widehat{G}$ be the initial states of $\widehat{G}_w$, $Obs_w(\widehat{G})$ is the variant observer of $\widehat{Obs}(\widehat{G})$ as defined in (\ref{eq:8}) ($Obs_w(\widehat{G})$ is not an observer of $\widehat{G}_w$), and the set of initial states $X_{0,c}^w=\bigcup_{\mycom{\hat{x}\in \hat{x}_{obs}} {\emptyset\neq \hat{x}_{obs} \in \widetilde{X}_{obs}}} \{(\widehat{x}, AS_w(\widehat{x}, \widehat{x}_{obs}))\}$, where $\widetilde{X}_{obs}$ is the set of reachable states in $\widehat{Obs}(\widehat{G})$.
For any state $x_c^w\in X_c^w$, the left component of $x_c^w$ tracks the states in $\widehat{G}$ while the right component of $x_c^w$ reflects the evolution of corresponding adjoint state set.
For any string $s_c\in \mathcal{L}(Cc(\widehat{G}_w, Obs_w(\widehat{G})))$, $L(s_c)$ and $R(s_c)$ are used to denote the left and right components of $s_c$, respectively.
Thus, we have $P^n(L(s_c)) = R(s_c)$.
Now, we use an example to show the construction of the concurrent composition $Cc(\widehat{G}_w, Obs_w(\widehat{G}))$.

\begin{example}
\rm{Consider again the system $G$ in Fig. \ref{fig:3.1}.
$\widehat{G}$ and its observer $\widehat{Obs}(\widehat{G})$ are shown in Figs. \ref{fig:4.5} and \ref{fig:4.6}, respectively\footnote{Note that we denote the state $(x, \widetilde{\sigma})\in \widehat{X}$ by $x \widetilde{\sigma}$ for the sake of convenience, where $\widetilde{\sigma}\in \{\epsilon\}\cup \bigcup_{i=1}^{|\Gamma|} \{\{(\sigma^1)_1, \dots, (\sigma^i)_i\}| \sigma^1, \dots, \sigma^i\in \Gamma \mbox{ are distinct}\}$. 
For example, states $(0, \epsilon)$ and $(2, a_1)$ are denoted by $0\epsilon$ and $2a_1$, respectively.}.
By the constructed $\widehat{Obs}(\widehat{G})$, the table of states and their corresponding adjoint state sets are shown in Table \ref{tab:4.2}.
The variant observer $Obs_w(\widehat{G})$ and variant automaton $\widehat{G}_w$ are depicted in Figs. \ref{fig:4.7} and \ref{fig:4.8}, respectively.
Then, the concurrent composition $Cc(\widehat{G}_w, Obs_w(\widehat{G}))$ is constructed in Fig. \ref{fig:4.9}.
}
$\hfill\square$
\end{example}

\begin{figure}[htp]
\centering
\includegraphics[width=2.8in]{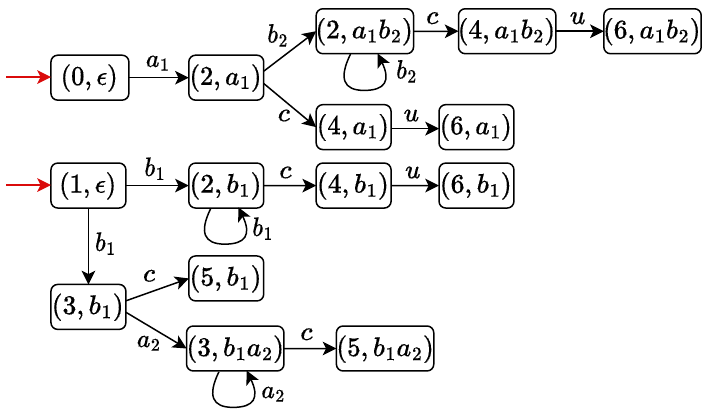}
\caption{The augmented automaton $\widehat{G}$ for $G$.}
\label{fig:4.5}
\end{figure}

\begin{figure}[htp]
\centering
\includegraphics[width=2.5in]{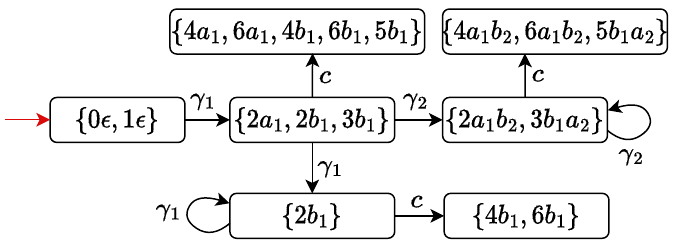}
\caption{Part of the observer $\widehat{Obs}(\widehat{G})$ for $\widehat{G}$.}
\label{fig:4.6}
\end{figure}

\begin{table}[htbp]
\renewcommand{\arraystretch}{1.2}
\caption{States and the corresponding adjoint state sets}
\label{tab:4.2}
\centering
\resizebox{7.2cm}{!}{
\begin{tabular}{ccc}
\toprule  
$x_{obs}$ & state & adjoint state set\\
\midrule
\multirow{2}{*}{$\{0\epsilon, 1\epsilon\}$} &$0\epsilon$ & $\{1\epsilon\}$ \\
&$1\epsilon$ & $\{0\epsilon\}$ \\
\midrule
\multirow{3}{*}{$\{2a_1, 2b_1, 3b_1\}$} & $2a_1$ & $\{3b_1\}$ \\
& $2b_1$ & $\{3b_1\}$ \\
& $3b_1$ & $\{2a_1, 2b_1\}$ \\
\midrule
$\{2b_1\}$ & $2b_1$ & $\emptyset$ \\
\midrule
\multirow{2}{*}{$\{4b_1,6b_1\}$} & $4b_1$ & $\{6b_1\}$ \\
& $6b_1$ & $\{4b_1, 6b_1\}$ \\
\midrule
\multirow{2}{*}{$\{2a_1b_2,3b_1a_2\}$} & $2a_1b_2$ & $\{3b_1a_2\}$ \\
& $3b_1a_2$ & $\{2a_1b_2\}$ \\
\midrule
\multirow{3}{*}{$\{4a_1b_2, 6a_1b_2, 5b_1a_2\}$} & $4a_1b_2$ & $\{6a_1b_2, 5b_1a_2\}$ \\
& $6a_1b_2$ & $\{4a_1b_2, 6a_1b_2, 5b_1a_2\}$ \\
& $5b_1a_2$ & $\{4a_1b_2, 6a_1b_2\}$ \\
\midrule
\multirow{5}{*}{$\{4a_1, 6a_1, 4b_1, 6b_1, 5b_1\}$} & $4a_1$ & $\{6a_1, 6b_1,5b_1\}$ \\
& $6a_1$ & $\{4a_1, 6a_1, 4b_1, 6b_1, 5b_1\}$ \\
& $4b_1$ & $\{6a_1, 6b_1, 5b_1\}$ \\
& $6b_1$ & $\{4a_1, 6a_1, 4b_1, 6b_1, 5b_1\}$ \\
& $5b_1$ & $\{4a_1, 6a_1, 4b_1, 6b_1\}$ \\
\bottomrule
\end{tabular}}
\end{table}

\begin{figure}[htp]
\centering
\includegraphics[width=3.3in]{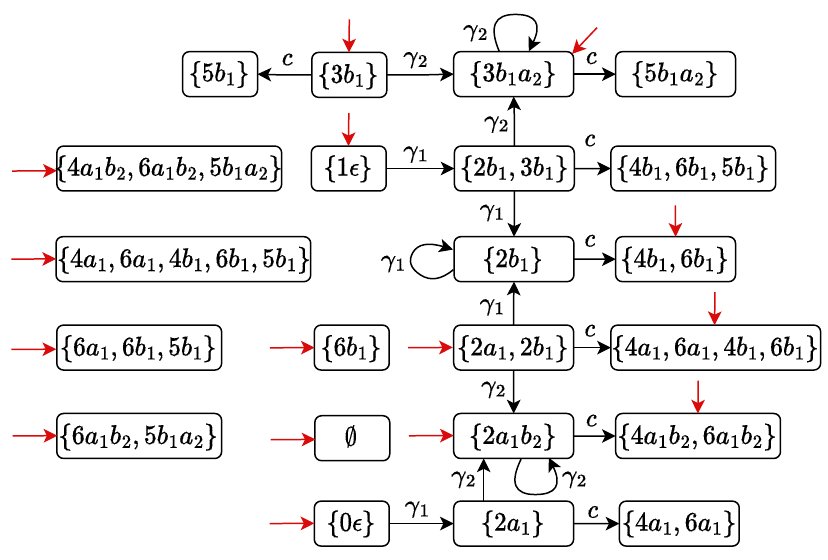}
\caption{Part of the variant observer $Obs_w(\widehat{G})$.}
\label{fig:4.7}
\end{figure}

\begin{figure}[htp]
\centering
\includegraphics[width=2.6in]{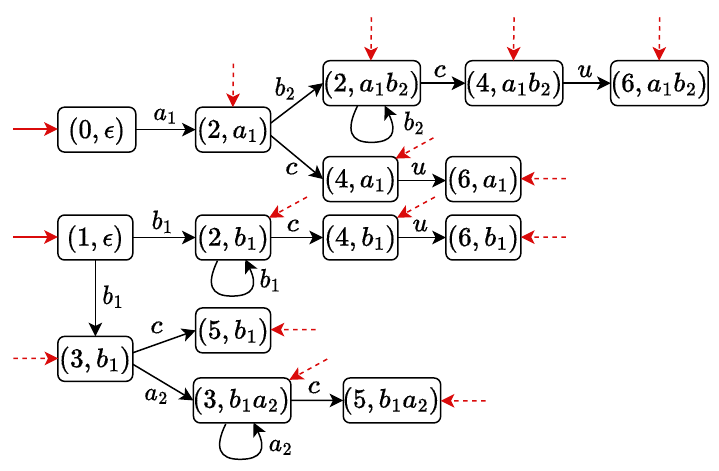}
\caption{The variant automaton $\widehat{G}_w$ for $\widehat{G}$.}
\label{fig:4.8}
\end{figure}

\begin{figure}[htp]
\centering
\includegraphics[width=3.1in]{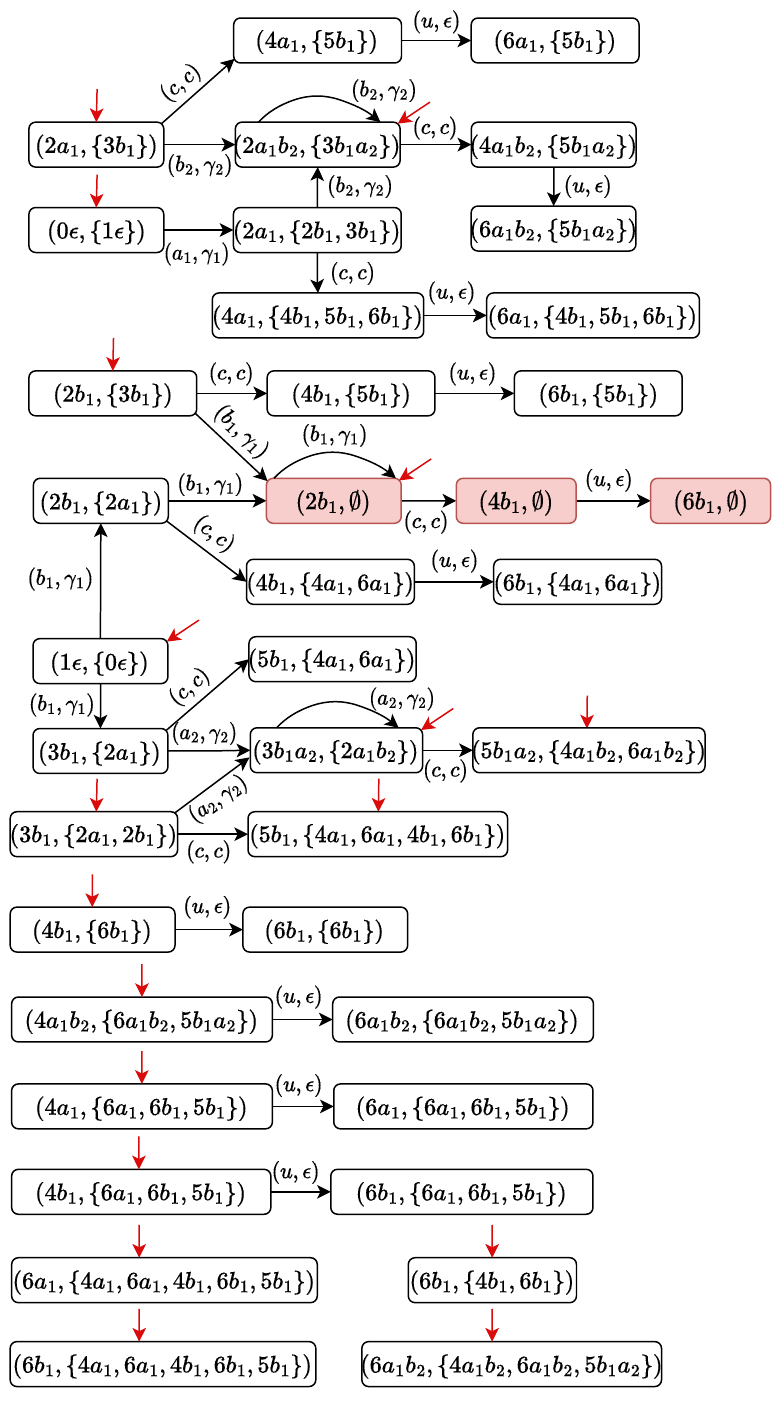}
\caption{Part of the concurrent composition $Cc(\widehat{G}_w,Obs_w(\widehat{G}))$.}
\label{fig:4.9}
\end{figure}

We present the main results of verifying $K$- and infinite-step weak anonymity using the proposed concurrent composition $Cc(\widehat{G}_w, Obs_w(\widehat{G}))$.

\begin{theorem}\label{th:4.2}
\rm{
Given a system $G=(X, \Sigma, \delta, X_0)$, let $\widehat{Obs}(\widehat{G}) = (\widehat{X}_{obs}, \widehat{\Sigma}_{obs}, \widehat{\delta}_{obs}, \widehat{X}_{0,obs})$ be the observer and $Cc(\widehat{G}_w, Obs_w(\widehat{G})) = (X_c^w, \Sigma_c^w, \delta_c^w, X_{0,c}^w)$ be the corresponding concurrent composition.
$G$ is infinite-step (resp., $K$-step) weakly anonymous w.r.t. $\Sigma_o$, $\Gamma$, and $P_w$ if and only if there exists no state of the form $(\cdot, \emptyset)$ in $Cc(\widehat{G}_w, Obs_w(\widehat{G}))$ that is reachable from $X_{0,c}^w$ (resp., within $K$ observational steps).
}
\end{theorem}
\begin{proof}
\rm{We first prove the $K$-step weak anonymity.\\
(if) By contrapositive, assume that system $G$ does not satisfy $K$-step weak anonymity w.r.t. $\Sigma_o$, $\Gamma$, and $P_w$.
By Definition \ref{def:3.3}, there exists a string $s't'\in \mathcal{L}(G)$ such that $|\widehat{E}_{G}(P_w, \alpha|\alpha\beta)|>1$ does not hold, where $P_w(s')=\alpha$, $P_w(s't')=\alpha\beta$, and $|\beta|\leq K$.
This implies $|\widehat{E}_{G}(P_w, \alpha|\alpha\beta)|=1$.
In this case, assume $\widehat{E}_{G}(P_w, \alpha|\alpha\beta)=\{x_w\}$.
By the construction of $\widehat{Obs}(\widehat{G})$, $(x_w, \cdot)\in \widehat{x}_{obs}$ can be obtained, where $\widehat{x}_{obs}=\widehat{\delta}_{obs}(\widehat{X}_{0,obs}, \alpha)$.
By the construction of $Obs_w(\widehat{G})$, one sees $\delta_w(AS_w((x_w, \cdot), \widehat{x}_{obs}), \beta)=\emptyset$.
By the construction of $Cc(\widehat{G}_w, Obs_w(\widehat{G}))$, there exists a string $s_c\in \mathcal{L}(Cc(\widehat{G}_w, Obs_w(\widehat{G})))$ such that $(x_{w,k} ,\emptyset)\in \delta_c^w(((x_w, \cdot), AS_w((x_w, \cdot), \widehat{x}_{obs})),s_c)$, where $x_{w,k}\in \widehat{\delta}((x_w, \cdot), L(s_c))$.
Due to $|P^n(L(s_c))| =|R(s_c)|\leq K$, there exists a state of the form $(\cdot, \emptyset)$ in $Cc(\widehat{G}_w, Obs_w(\widehat{G}))$ that is reachable from $X_{0,c}^w$ within $K$ observational steps.\\
(only if) Still by contrapositive, suppose that there exists a state of the form $(\cdot, \emptyset)$ in $Cc(\widehat{G}_w, Obs_w(\widehat{G}))$ that is reachable from $X_{0,c}^w$ within $K$ observational steps.
By the construction of $Cc(\widehat{G}_w, Obs_w(\widehat{G}))$, one concludes that there exists a state $(x_{w,k}, \emptyset)\in X_c^w$ reached from some initial state $((x_w, \cdot), AS_w((x_w, \cdot), \widehat{x}_{obs}))\in X_{0,c}^w$ and a string $s_c\in \mathcal{L}(Cc(\widehat{G}_w, Obs_w(\widehat{G})), ((x_w, \cdot), AS_w((x_w, \cdot), \widehat{x}_{obs})))$ with $|R(s_c)|\leq K$, where $(x_w, \cdot)\in \widehat{x}_{obs}$, $x_{w,k}\in \widehat{\delta}((x_w, \cdot), L(s_c))$, and $(x_{w,k}, \emptyset)\in \delta_c^w(((x_w, \cdot), AS_w((x_w, \cdot),\\ \widehat{x}_{obs})), s_c)$.
By the construction of $Obs_w(\widehat{G})$, we infer that $\delta_w(AS_w((x_w, \cdot), \widehat{x}_{obs}), R(s_c))=\emptyset$.
Further, $\widehat{\delta}(x'_w, t)=\emptyset$ holds for all $x'_w\in AS_w((x_w, \cdot), \widehat{x}_{obs})$ and all $t\in \widehat{\Sigma}^*$ with $P^n(t)=R(s_c)$.
Then, by the construction of $\widehat{Obs}(\widehat{G})$, there exists an initial state $\widehat{x}_0\in \widehat{X}_0$ and a string $s\in \mathcal{L}(\widehat{G}, \widehat{x}_0)$ such that $(x_w, \cdot) \in \widehat{\delta}(\widehat{x}_0, s)$ and $\widehat{\delta}_{obs}(\widehat{X}_{0,obs}, P^n(s)) =\widehat{x}_{obs}$ due to $(x_w, \cdot)\in \widehat{x}_{obs}$.
Let $P^n(s)=\alpha$ and $P^n(t)=\beta$.
Based on the above analysis, $x_w$ is the unique state in $\widehat{E}_{G}(P_w, \alpha|\alpha\beta)$, i.e., $|\widehat{E}_{G}(P_w, \alpha|\alpha\beta)|=1$, where $|\beta|\leq K$.
System $G$ does not satisfy $K$-step weak anonymity w.r.t. $\Sigma_o$, $\Gamma$, and $P_w$.

The proof of the case of infinite-step weak anonymity can be obtained similarly.
}
\end{proof}

We propose an algorithm for verifying $K$- and infinite-step weak anonymity, as shown in Algorithm \ref{al:2}.
As seen, the time complexity of computing $\widehat{G}$, $\widehat{Obs}(\widehat{G})$, $Obs_w(\widehat{G})$, $\widehat{G}_w$, and $Cc(\widehat{G}_w, Obs_w(\widehat{G}))$ are 
$\mathcal{O}((|\Gamma|^2+|\Sigma_o \backslash \Gamma|+|\Sigma_{uo}|) |\widehat{X}|^2)$, 
$\mathcal{O}(|\Gamma| + |\Sigma_o\backslash \Gamma|) |\Sigma_{uo}| |\widehat{X}|^2 2^{|\widehat{X}|})$, 
$\mathcal{O}(|\Gamma| + |\Sigma_o\backslash \Gamma|) |\Sigma_{uo}| |\widehat{X}|^2 2^{|\widehat{X}|})$, 
$\mathcal{O}((|\Gamma|^2+|\Sigma_o \backslash \Gamma|+|\Sigma_{uo}|) |\widehat{X}|^2)$,
and $\mathcal{O}((|\Gamma|^2+|\Sigma_o\backslash\Gamma|+|\Sigma_{uo}|) |\widehat{X}|^2 2^{|\widehat{X}|})$, respectively, where $|\widehat{X}|=|X|\sum_{k=0}^{|\Gamma|}\frac{|\Gamma|!}{k!}$.
The overall time complexity of verifying $K$- and infinite-step weak anonymity is 
$\mathcal{O}((2|\Sigma_{uo}|(|\Gamma|+|\Sigma_o\backslash\Gamma|) + (|\Gamma|^2+|\Sigma_o\backslash\Gamma|+|\Sigma_{uo}|)) |\widehat{X}|^2 2^{|\widehat{X}|})$.

\begin{algorithm}[!htbp]
  \DontPrintSemicolon
  \caption{Verification of Infinite-Step (resp., $K$-step) Weak Anonymity}\label{al:2}
  \KwIn{System $G=(X, \Sigma, \delta, X_0), \Sigma_o, \Sigma_{uo}, \Gamma \subseteq \Sigma_o$, and $K\in \mathbb{N}$.}
  \KwOut{``YES" if $G$ satisfies infinite-step (resp., $K$-step) weak anonymity w.r.t. $\Sigma_o$, $\Gamma$, and $P_w$, ``NO" otherwise.}
   Construct the augmented automaton $\widehat{G}$ for $G$;\;
   Construct the observer $\widehat{Obs}(\widehat{G})$ for $\widehat{G}$;\;
  \eIf{there exists a state $\widehat{x}_{obs}\in \widehat{X}_{obs}$ such that $|m(\widehat{x}_{obs})|=1$,}
  {\textbf{return} ``No'';\;
	\textbf{stop};\;
  }{
	Construct modified observer $Obs_w(\widehat{G})$ for $\widehat{G}$;\;
	Construct $\widehat{G}_w$ as a variant of $\widehat{G}$;\;
	Construct $Cc(\widehat{G}_w, Obs_w(\widehat{G}))$ as the concurrent composition of $\widehat{G}_w$ and $Obs_w(\widehat{G})$;\;
	Use the ``Breadth-First Search Algorithm” to decide whether there exists a state of the form $(\cdot,\emptyset)$ in $Cc(\widehat{G}_w, Obs_w(\widehat{G}))$ that is reachable from $X_{0,c}^w$ (resp., within $K$ observational steps);\;
	\eIf{there exists such a state $(\cdot,\emptyset)$ in $Cc(\widehat{G}_w, Obs_w(\widehat{G}))$ that is reachable from $X_{0,c}$ (resp., within $K$ observational steps),}
	{\textbf{return} ``No'';\;
	\textbf{stop};\;
	}
	{\textbf{return} ``Yes'';\;
	\textbf{stop}.
	}
  }
\end{algorithm}

\begin{example}
\rm{
Reconsider Example \ref{ex:3.3}.
The corresponding concurrent composition $Cc(\widehat{G}_w, Obs_w(\widehat{G}))$ is shown in Fig. \ref{fig:4.9}, where there exists an initial state $(2b_1, \emptyset)$.
By Theorem \ref{th:4.2},
we conclude that $G$ is not $K$-step weakly anonymous w.r.t. to $P_w$ and $\Gamma$ for any $K$.
Therefore, $G$ is also not infinite-step weakly anonymous w.r.t. to $P_w$ and $\Gamma$.
The results are consistent with Example \ref{ex:3.3}.
}
$\hfill\square$
\end{example}

Finally, we discuss both the proposed verification method and the existing ones for verifying anonymity and opacity properties.
Two comparison tables of methods for verifying different types of anonymity and opacity with complexities are presented in Tables \ref{tab:4.3} and \ref{tab:4.4}, respectively.
From the tables, one sees that our concurrent composition plus the observer method currently provides the most efficient algorithms for verifying all kinds of anonymity and opacity properties. 
Furthermore, we find that the complexities of methods for verifying $K$- and infinite-step opacity and anonymity (i.e., initial-state estimator, two-way observer, and non-secret specification plus the observer) all depend on the value of $K$.
In contrast, the idea of the concurrent-composition technique is to construct a $K$-free information structure.
It means that once a system is given, a concurrent composition information structure can be constructed to verify $K$-step strong (resp., weak) anonymity for any $K\in \mathbb{N}$.
This contrasts with the two-way observer approach, where the construction depends explicitly on $K$ (i.e., for each different $K$, a new two-way observer must be constructed).
To sum up, the use of concurrent-composition techniques results in lower complexity and enables the constructed information structure to be independent of $K$.

Note that the proposed method for verifying $K$- and infinite-step strong anonymity can also be used to verify $K$- and infinite-step anonymity in \cite{Basilio2021}.
In addition, the methods listed in Table \ref{tab:4.4} for verifying $K$- and infinite-step opacity can also be used to verify $K$- and infinite-step anonymity (resp., $K$- and infinite-step strong anonymity with slight modifications).
When $\Gamma=\emptyset$, $K$-step (resp., infinite-step) strong anonymity and weak anonymity reduce to $K$-step (resp., infinite-step) anonymity.
In this case, all methods for verifying $K$- and infinite-step opacity can be used for their verification.

\begin{table*}[htbp]
\renewcommand{\arraystretch}{1.3}
\caption{Methods for verifying different types of anonymity with complexities}
\label{tab:4.3}
\centering
\resizebox{17cm}{!}{
\begin{tabular*}{1.08\linewidth}{@{}lllll@{}}
\toprule
     \begin{tabular}{l}
     method 
     \end{tabular} &
     \begin{tabular}{l}
     current-state 
     anonymity \end{tabular} &
     \begin{tabular}{l}
     $K$/infinite-step 
     anonymity \end{tabular} & 
     \begin{tabular}{l}
     $K$/infinite-step strong\\
     anonymity \end{tabular} & 
     \begin{tabular}{l}
     $K$/infinite-step weak\\
     anonymity \end{tabular} \\
\midrule
 \begin{tabular}{l}
 observer
 \end{tabular} &
 \begin{tabular}{l}
 $\mathcal{O}(|\Sigma_o||\Sigma||X|^2 2^{|X|})$
 (\cite{Wu2014})
 \end{tabular} & & & \\
\midrule
 \begin{tabular}{l}
 two-way
 observer \end{tabular} & &
 \begin{tabular}{l}
 $\mathcal{O}(|\Sigma_o||\Sigma||X|^2 2^{|X|+1}+$\\
 min\{$2^{|X|}, {|\Sigma_o|}^K$\}$|\Sigma_o| 2^{|X|})$/\\
 $\mathcal{O}(|\Sigma_o||\Sigma||X|^2 2^{|X|+1}+$\\
 $|\Sigma_o| 4^{|X|})$ (\cite{Basilio2021}, \cite{Yin2017})
 \end{tabular} & 
 \begin{tabular}{l}
 $\mathcal{O}(|\Sigma_o||\Sigma||X|^2 2^{|X|+1}+$\\
 min\{$2^{|X|}, {|\Sigma_o|}^K$\}$|\Sigma_o| 2^{|X|})$/\\
 $\mathcal{O}(|\Sigma_o||\Sigma||X|^2 2^{|X|+1}+$\\
 $|\Sigma_o| 4^{|X|})$ (\cite{Basilio2021}, \cite{Yin2017})
 \end{tabular} & \\
\midrule
 \begin{tabular}{l}
 concurrent composition\\
 + observer
 \end{tabular} & 
 \begin{tabular}{l}
 $\mathcal{O}(|\Sigma_o||\Sigma||X|^2 2^{|X|})$\\
 (\cite{Kuize2023} with slight\\
 modification)
 \end{tabular} &
 \begin{tabular}{l}
 $\mathcal{O}(|X|^2 2^{|X|} (|\Sigma_o||\Sigma_{uo}|+$\\
 $|\Sigma|))$\\
 (Theorem \ref{th:4.1} with slight\\
 modification)
 \end{tabular} & 
 \begin{tabular}{l}
 $\mathcal{O}(|X|^2 2^{|X|} (|\Sigma_o||\Sigma_{uo}|+$\\
 $|\Sigma|))$\\
 (Theorem \ref{th:4.1})
 \end{tabular} & 
 \begin{tabular}{l}
 $\mathcal{O}(|\widehat{X}|^2 2^{|\widehat{X}|} (2|\Sigma_o|(|\Gamma|+$\\
 $|\Sigma_o\backslash \Gamma|)+(|\Gamma|^2+$\\
 $|\Sigma_o\backslash \Gamma|+|\Sigma_{uo}|))$\\
 (Theorem \ref{th:4.2})
 \end{tabular} \\
\bottomrule
\end{tabular*}}
\end{table*}

\begin{table*}[htbp]
\renewcommand{\arraystretch}{1.3}
\caption{Methods for verifying different types of state-based opacity with complexities}
\label{tab:4.4}
\centering
\resizebox{17cm}{!}{
\begin{tabular*}{1.06\linewidth}{@{}lllll@{}}
\toprule
     \begin{tabular}{l}
     method 
     \end{tabular} &
     \begin{tabular}{l}
     current-state 
     opacity \end{tabular} &
     \begin{tabular}{l}
     initial-state
     opacity \end{tabular} & 
     \begin{tabular}{l}
     $K$-step 
     opacity \end{tabular} & 
     \begin{tabular}{l}
     infinite-step 
     opacity \end{tabular} \\
\midrule
 \begin{tabular}{l}
 observer
 \end{tabular} &
 \begin{tabular}{l}
 $\mathcal{O}(|\Sigma_o||\Sigma||X|^2 2^{|X|})$ \\
 (\cite{Saboori2007})
 \end{tabular} & & & \\
\midrule
 \begin{tabular}{l}
 reverse observer
 \end{tabular} &  & 
 \begin{tabular}{l}
 $\mathcal{O}(|\Sigma_o||\Sigma||X|^2 2^{|X|})$ (\cite{Wu2013})
 \end{tabular} & & \\
\midrule
 \begin{tabular}{l}
 initial-state
 estimator \end{tabular} &  & 
 \begin{tabular}{l}
 $\mathcal{O}(|\Sigma_o|2^{|X|^2})$ (\cite{Wu2013})
 \end{tabular} & 
 \begin{tabular}{l}
 $\mathcal{O}(|\Sigma_o| 2^{|X|} (|\Sigma_o|+$ \\
 $1)^{\mbox{min}(K, 2^{|X|^2-1})})$ (\cite{Saboori2011})
 \end{tabular}& 
 \begin{tabular}{l}
 $\mathcal{O}(|\Sigma_o|2^{|X|+|X|^2})$ (\cite{Saboori2012})
 \end{tabular} \\
\midrule
 \begin{tabular}{l}
 two-way
 observer \end{tabular} & & & 
 \begin{tabular}{l}
 $\mathcal{O}(|\Sigma_o||\Sigma||X|^2 2^{|X|+1}+$\\
 min\{$2^{|X|}, {|\Sigma_o|}^K$\}$|\Sigma_o| 2^{|X|})$ \\(\cite{Yin2017})
 \end{tabular} & 
 \begin{tabular}{l}
 $\mathcal{O}(|\Sigma_o||\Sigma||X|^2 2^{|X|+1}+$\\
 $|\Sigma_o| 4^{|X|})$ (\cite{Yin2017})
 \end{tabular} \\
 \midrule
 \begin{tabular}{l}
 non-secret 
 specification\\
 automaton
 + observer \end{tabular} & 
 \begin{tabular}{l}
 $\mathcal{O}(|\Sigma_o||\Sigma||X|^2 2^{2|X|+3})$\\
 (\cite{Wintenberg2022})
 \end{tabular} & 
 \begin{tabular}{l}
 $\mathcal{O}(3|\Sigma_o||\Sigma||X|^2 2^{3|X|+2})$\\
 (\cite{Wintenberg2022})
 \end{tabular} & 
 \begin{tabular}{l}
 $\mathcal{O}(|\Sigma_o||\Sigma||X|^2 (K+$ \\
 $3) 2^{(K+3)|X|+2})$ (\cite{Wintenberg2022})
 \end{tabular}& 
 \begin{tabular}{l}
 $\mathcal{O}(|\Sigma_o||\Sigma||X|^2 (2^{|X|}+$ \\
 $3) 2^{(2^{|X|}+3)|X|+2})$ (\cite{Wintenberg2022})
 \end{tabular} \\
\midrule
 \begin{tabular}{l}
 concurrent composition\\
 + observer
 \end{tabular} & 
 \begin{tabular}{l}
 $\mathcal{O}(|\Sigma_o||\Sigma||X|^2 2^{|X|})$ \\
 (\cite{Kuize2023})
 \end{tabular} &
 \begin{tabular}{l}
 $\mathcal{O}(|\Sigma_o||\Sigma||X|^2 2^{|X|})$ \\
 (\cite{Kuize2023})
 \end{tabular} & 
 \begin{tabular}{l}
 $\mathcal{O}(|\Sigma_o||\Sigma||X|^2 2^{|X|})$ \\
 (\cite{Balun2023}, called product\\
 automaton)
 \end{tabular} & 
 \begin{tabular}{l}
 $\mathcal{O}(|\Sigma_o||\Sigma||X|^2 2^{|X|})$ \\
 (\cite{Balun2021}, called product\\
 automaton)
 \end{tabular} \\
\bottomrule
\end{tabular*}}
\end{table*}

\section{Maximal K and upper bounds on K for K-step strong/weak anonymity}\label{sec:5}
\subsection{Maximal K for K-step strong/weak anonymity}
In this subsection, we aim to obtain the maximal value of $K$ that makes system $G$ satisfy $K$-step strong (resp., weak) anonymity.
To this end, we first define the minimum length of the strong (resp., weak) exposure distance $k_s$ (resp., $k_w$) which is the minimum value of $K$ that violates $K$-step strong (resp., weak) anonymity.
\begin{definition}\label{de:5.1}(Minimum length of the strong and weak exposure distance)
\rm{Given a system $G$, its corresponding two concurrent compositions are $Cc(G_m, Obs_m(G))  = (X_c, \Sigma_c, \delta_c, X_{0,c})$ and $Cc(\widehat{G}_w, Obs_w(\widehat{G})) = (X_c^w, \Sigma_c^w, \delta_c^w,\\ X_{0,c}^w)$, respectively.\\
(1) A non-negative integer $k_s\in \mathbb{N}$ is said to be the minimum length of the strong exposure distance if
\begin{align*}
(\nexists x_{0, c}\in X_{0,c}) (\nexists t\in \mathcal{L}(Cc(G_m, Obs_m(G)), x_{0, c}))
(\nexists Z\in \mathbb{N})\\
[(\cdot, \emptyset)\in \delta_c(x_{0, c}, t) \land |R(t)|=Z \land Z < k_s].
\end{align*}
(2) A non-negative integer $k_w\in \mathbb{N}$ is said to be the minimum length of the weak exposure distance if
\begin{align*}
(\nexists x_{0, c}^w\in X_{0,c}^w) (\nexists t'\in \mathcal{L}(Cc(\widehat{G}_w, Obs_w(\widehat{G})), x_{0, c}^w))
(\nexists Z'\in \mathbb{N})\\
[(\cdot, \emptyset)\in \delta_c^w(x_{0, c}^w, t') \land |R(t')|=Z' \land Z' < k_w].\rightdiamond{}
\end{align*}}
\end{definition}
Now, we design an algorithm to compute the maximal $K$ for $K$-step strong/weak anonymity, as shown in Algorithm \ref{al:3}.
In $Cc(G_m, Obs_m(G))$ (resp., $Cc(\widehat{G}_w, Obs_w(\widehat{G}))$), we check whether there exists a state taking the form $(\cdot, \emptyset)$ reached from an initial state in $X_{0,c}$ (resp., $X_{0,c}^w$).
If the answer is ``no", Algorithm \ref{al:3} outputs ``$\infty$", which means that system $G$ satisfies infinite-step strong (resp., weak) anonymity.
If the answer is ``yes", we need to find the minimum length of the strong (resp., weak) exposure distance $k_s$ (resp., $k_w$).
Algorithm \ref{al:3} outputs ``NA" if $k_s = 0$ (resp., $k_w = 0$), indicating that the maximal $K$ does not exist.
Algorithm \ref{al:3} outputs ``$k_s-1$" (resp., ``$k_w-1$") if $k_s > 0$ (resp., $k_w > 0$), suggesting that the maximal $K$ for $K$-step strong (resp., weak) anonymity is $k_s-1$.

\begin{algorithm}[!htbp]
  \DontPrintSemicolon
  \caption{Computing the Maximal $K$ for $K$-Step Strong (resp., Weak) Anonymity}\label{al:3}
  \KwIn{System $G=(X, \Sigma, \delta, X_0), \Sigma_o, \Sigma_{uo}$, and $\Gamma \subseteq \Sigma_o$.}
  Construct the corresponding concurrent composition $Cc(G_m, Obs_m(G))$ (resp., $Cc(\widehat{G}_w, Obs_w(\widehat{G}))$);\;
  \eIf{there exists no state of the form $(\cdot, \emptyset)$ in $Cc(G_m, Obs_m(G))$ (resp., $Cc(\widehat{G}_w, Obs_w(\widehat{G})))$}
  {Output ``$\infty$" and exit;\;
  }{
     Use the ``Breadth-First Search” to find the minimum length of the strong (resp., weak) exposure distance $k_s$ (resp., $k_w$);\;
	\eIf{$k_s = 0$ (resp., $k_w = 0$)}
	{Output ``NA" and exit;\;
	}
	{Output ``$k_s-1$" (resp., ``$k_w-1$") and exit;
	}
  }
\end{algorithm}


\subsection{Upper bounds on K in K-step strong anonymity and weak anonymity}
In this subsection, we will give two upper bounds on $K$ for $K$-step strong anonymity and weak anonymity, respectively.

\begin{proposition}\label{pro:6.1}
\rm{Given a system $G=(X, \Sigma, \delta, X_0)$, for any $\widehat{K} > K \geq |X|\cdot 2^{|X|}$ (resp., $|\widehat{X}|\cdot 2^{|\widehat{X}|}$), $G$ is $K$-step strongly (resp., weakly) anonymous w.r.t. $\Sigma_o$, $\Gamma$, and $P_{str}$ (resp., $P_w$) if and only if $G$ is $\widehat{K}$-step strongly (resp., weakly) anonymous w.r.t. $\Sigma_o$, $\Gamma$, and $P_{str}$ (resp., $P_w$).}
\end{proposition}
\begin{proof}
\rm{We show it holds for the case of $K$-step strong anonymity first.
By Proposition \ref{pro:3.1}, the ``if" part holds.
We focus on the ``only if" part, which requires us to prove that if $G$ satisfies $K$-step strong anonymity, then $G$ satisfies $\widehat{K}$-step strong anonymity.
Without loss of generality, we assume $\widehat{K} = K + 1$, which can be naturally generalized to $\widehat{K} > K + 1$.
For brevity, let $e=(|X| \cdot 2^{|X|}) + 1$.
Then, we only prove that $G$ is $e$-step strongly anonymous if $G$ is $(e-1)$-step strongly anonymous.
By contrapositive, suppose that system $G$ does not satisfy $e$-step strong anonymity.
It means that there exists a string $st\in \mathcal{L}(G)$ with $P_{str}(s)=\alpha$, $P_{str}(t)=\beta$, and $|\beta|\leq e$ such that $|\widehat{E}_{G}(P_{str}, \alpha|\alpha\beta)|>1$ does not hold.
By the construction of $Cc(G_m, Obs_m(G))$, there exists a state $x_{n,c} = (\cdot, \emptyset)$ reached by a string $s_c\in \mathcal{L}(Cc(G_m, Obs_m(G)))$ from an initial state $x_{0, c}$ within $e$ observational steps, i.e., $x_{n,c} \in \delta_c(x_{0, c}, s_c)$ and $|P_{str}(L(s_c))| = |R(s_c)| \leq e$.
Consider the complete run $x_{0,c}\xrightarrow{(\sigma_1, \eta_1)} x_{1,c}\xrightarrow{(\sigma_2, \eta_2)} \cdots \xrightarrow{(\sigma_n, \eta_n)} x_{n,c}=(\cdot, \emptyset)$.
Let the actual number of observational steps in $s_c$ be $k$, i.e., $k=|R(s_c)|=|\eta_1\eta_2\cdots \eta_n|$, where $k\leq e$.
Now, identify the positions where new observational steps occur. 
Let $0=m_0<m_1<m_2<\cdots <m_k=n$ be the indices where $\eta_{m_i}\neq \epsilon$.
Define the immediate state after the $i$-th observational step as $y_{i,c}=x_{m_i,c}$.
One has $y_{k,c}=x_{m_k,c}=x_{n,c}=(\cdot, \emptyset)$ since the right component $\emptyset$ of $x_{n,c}$ is derived from the last updated observational step at the end of the string $s_c$.
The state space of $Cc(G_m, Obs_m(G))$ has at most $N=|X|\cdot 2^{|X|}$ states. 
Now, we need to consider two cases:
Case 1: $k\leq N$.
In this case state $y_{k,c}=(\cdot,\emptyset)$ is reached within $k\leq N\leq e-1$ observational steps, which contradicts the $(e-1)$-step strong anonymity.
Case 2: $k = N+1 = e$.
In this case there are $e+1$ states: $y_{0,c}, y_{1,c}, \dots$, and $y_{e,c}$, where $y_{e,c}=(\cdot, \emptyset)$.
By the pigeonhole principle, there exist indices 
$i$ and $j$ with $0\leq i<j\leq e$ such that $y_{i,c}=y_{j,c}$.
Let $\overline{s}_c$ be the prefix of $s_c$ that ends at state $y_{i,c}$ (after $i$-th observational steps). 
State $y_{i,c}=(\cdot,\emptyset)$ is reached through the string $\overline{s}_c$ within $i\leq e-1$ observational steps, contradicting the $(e-1)$-step strong anonymity.
Both cases show that system $G$ does not satisfy $(e-1)$-step strong anonymity and thus completes the proof.

The proof of the upper bounds for $K$-step weak anonymity is very similar to that of $K$-step strong anonymity.
We omit it here due to the limited space.}
\end{proof}

Theoretically, the meaning of the upper bound $\Bar{K}$ on $K$ in $K$-step strong or weak anonymity is the ``maximum depth" that the corresponding concurrent composition $Cc(G_m, Obs_m(G))$ or $Cc(\widehat{G}_w, Obs_w(\widehat{G}))$ needs to be examined.
In other words, after undergoing the corresponding $\Bar{K}$ observational steps, concurrent composition $Cc(G_m, Obs_m(G))$ or $Cc(\widehat{G}_w, Obs_w(\widehat{G}))$ no longer generates any new state that does not appear in paths of length less than or equal to $\Bar{K}$.
Moreover, for all values $L$ greater than $\Bar{K}$, $L$-step strong/weak anonymity property holds naturally.
That is, if a system satisfies (or does not satisfy) $K$-step strong/weak anonymity when $K$ reaches its upper bound $\Bar{K}$, then for any value $L$ greater than $\Bar{K}$, the system still satisfies (or does not satisfy) $L$-step strong/weak anonymity.

Practically, the upper bound $\Bar{K}$ on $K$ may be significant for dynamic security strategies.
For instance, cloud service providers aim to protect customers' virtual machines from side-channel attacks. 
The cloud platform defense system could consider utilizing $\Bar{K}$ to implement a dynamic security strategy.
Specifically, a virtual machine is migrated to a potentially insecure host machine. 
During $\Bar{K}$ observational steps, a high-intensity and high-overhead active defense mechanism is employed, such as real-time monitoring and/or obfuscation techniques.
This ensures the state non-uniqueness of the virtual machine, such that the cloud platform defense system satisfies $K$-step strong/weak anonymity.
Based on the above theoretical analysis, the virtual machine is stable on its current host and maintains $K$-step strong/weak anonymity when the number of observational steps exceeds $\Bar{K}$, which drives the system to switch to a low-cost maintenance mode.
In a word, a high-cost mode will stop until $\Bar{K}$ is reached, at which point a low-cost maintenance mode is activated. 
This strategy minimizes performance degradation caused by defensive measures such as migration and monitoring while ensuring security.
Notably, an optimal balance between safety and efficiency is achieved.

\begin{corollary}\label{co:6.1}
\rm{Given a system $G=(X, \Sigma, \delta, X_0)$, a strong (resp., weak) anonymous projection $P_{str}$ (resp., $P_w$) w.r.t. a set of anonymous events $\Gamma$ and a set of observable events $\Sigma_o$ with $\Gamma\subseteq \Sigma_o$, $G$ is infinite-step strongly (resp., weakly) anonymous w.r.t. $\Sigma_o$, $\Gamma$, and $P_{str}$ (resp., $P_w$) if and only if $G$ is $(|X|\cdot 2^{|X|})$-step strongly (resp., $(|\widehat{X}|\cdot 2^{|\widehat{X}|})$-step weakly) anonymous w.r.t. $\Sigma_o$, $\Gamma$, and $P_{str}$ (resp., $P_w$).}
\end{corollary}
\begin{proof}
\rm{The results are obtained directly from Propositions \ref{pro:3.1} and \ref{pro:6.1}.}
\end{proof}

\section{A case study on a simple anonymous questioning process}\label{sec:6}
In this section, we apply the developed notions of $K$- and infinite-step strong anonymity and weak anonymity to an anonymous questioning process to demonstrate the proposed methods.
We adapt the example from \cite{Bryans2008} and add more details.

A chemical development company R asks company Q to prepare a feasibility study on developing a new chemical.
When this procedure is completed, company R is informed of the result of this study.
Then, company R decides to commission a chemical safety report from a specialized agency.
The specialized agency will choose some professional consultants.
The relevant law allows the chosen professional consultants to question company Q on aspects of the feasibility study.
To ensure the integrity and impartiality of the answer of company Q, the identity of any exact consultant involved cannot be inferred by company Q.
In other words, from the viewpoint of company Q, the observable actions cannot expose any exact consultant.

We model a simple anonymous questioning process for company Q, as shown in Fig. \ref{fig:7.1}.
The whole process should start with an initial review stage before moving on to the questioning stage.
The green part of Fig. \ref{fig:7.1} represents that there are two consultants M and N who are responsible for the initial review process, whereas the blue part of Fig. \ref{fig:7.1} denotes the questioning process between consultants A--G and company Q.
The questioning of consultants is sequential since the latter's questioning is based on the former's questioning.
For instance, consultant A is only allowed to initiate questioning after the initial review of consultant M while consultant B is not.
The events $a-g$ represent that ``consultants A--G initiate questioning to company Q", respectively.
It is a fact that only the questioning process with the associated sensors can be observed.
In the general cases, not all questioning devices are equipped with sensors due to the consideration of the cost savings.
This results in the situation that only part of the questioning process can be observed.
Assume that the actions ``consultants A, B, C, D, and F initiate questioning to company Q" can be observed by company Q since company Q can detect the sensor signals from the questioning devices of consultants A, B, C, D, and F, i.e., $\Sigma_o=\{a,b,c,d,f\}$.
The occurrence of actions ``consultants E and G initiate questioning to company Q" cannot be captured by company Q since there are no sensors deployed on the questioning devices of consultants E and G and company Q cannot detect the sensor signals, i.e., $\Sigma_{uo}=\{e, g\}$.
To maximize the concealment of the identity of all consultants, 
consultants A, B, C, D, and F add anonymous processors to their questioning devices, i.e., $ \Gamma=\Sigma_o$.
States M and N represent that consultants M and N are initially reviewing company Q, respectively.
States A--G represent that company Q answers consultants A--G, respectively.

\begin{figure}[htp]
\centering
\includegraphics[width=2.6in]{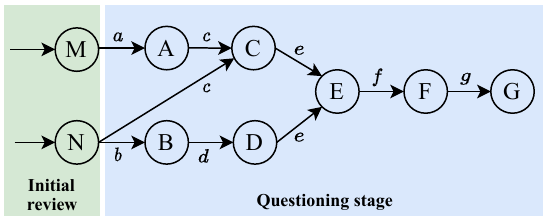}
\caption{System $G^*$: The model of a simple anonymous questioning process for company Q.}
\label{fig:7.1}
\end{figure}

We need to consider two cases for the processing power of anonymous processors.
One is that the processing power of anonymous processors is ideal, which means that actions representing ``consultants A, B, C, D, and F initiate questioning to company Q" can be observed as $\gamma$ by company Q.
The other is that the processing power of anonymous processors is less powerful such that the same result can be obtained for specific actions only performed by the same agent.
For instance, company Q can observe $\gamma_1\gamma_2$ when action string $ac$ or $bd$ happens.
Specifically,

\begin{itemize}
\item The first case corresponds to the strong anonymous projection.
The concurrent composition $Cc(G^*_m, Obs_m(G^*))$ for the modeled system $G^*$ can be constructed to verify $K$-step strong anonymity, which is omitted here to save space.
We conclude that the system is $K$-step strongly anonymous for any value of $K$, suggesting that the modeled system is infinite-step strongly anonymous.
For instance, assume $K=1$.
Consider action strings $s=ce$ and $t=fg$.
Then, $P_{str}(s)=\gamma$ and $P_{str}(st)=\gamma\gamma$ can be obtained.
The delayed state estimate at the instant that the first $\gamma$ has just been observed upon observing $\gamma\gamma$ can be inferred as $\{A, B, C, E\}$, i.e., $\widehat{E}_{G}(P_{str}, \gamma|\gamma\gamma)=\{A, B, C, E\}$.
It means that company Q cannot determine the identity of any precise consultant when the first $\gamma$ is captured even according to $1$ observational step future information if the processing power of anonymous processors is ideal such that company Q cannot distinguish questions from different consultants.
\item The second case corresponds to the weak anonymous projection.
Similarly, we verify $K$-step weak anonymity by constructing the concurrent composition $Cc(\widehat{G}^*_w, Obs_w(\widehat{G}^*))$ of the modeled system $G^*$, which is also omitted here.
By computation, it is found that the modeled system $G^*$ satisfies infinite-step weakly anonymity.
It means that company Q cannot infer the identity of any accurate consultant at any instant even though it knows future information when the processing power of anonymous processors is less powerful such that company Q can only distinct the actions performed by the same consultant.
\end{itemize}

\section{Concluding Remarks}\label{sec:7}
Anonymity and opacity are both information flow properties based on state estimates.
The four $K$- and infinite-step anonymity concepts proposed in this paper are essentially different from the existing $K$- and infinite-step anonymity \cite{Basilio2021, Yin2020} and opacity \cite{Yin2017}--\cite{Balun2023} since we take the strong and weak anonymous projections into account.
Particularly, when the anonymous events are not considered, i.e., $\Gamma=\emptyset$, the strong and weak anonymous projections reduce to the natural projection.
Note that the strong anonymous projection $P_{str}$ can be viewed as a natural extension of the natural projection, but the weak anonymous projection $P_w$ is more complicated, i.e., the weak anonymous projection of a string contains more information than its strong anonymous projection.
For instance, assume $\Gamma=\{a, b\}\subseteq \Sigma_o$.
The strong anonymous projection of strings $aab$ and $abb$ are $P_{str}(aab)=P_{str}(abb)=\gamma\gamma\gamma$, whereas the weak anonymous projection of strings $aab$ and $abb$ are $P_w(aab)=\gamma_1\gamma_1\gamma_2$ and $P_w(abb)=\gamma_1\gamma_2\gamma_2$.
That is, the weak anonymous projection of an anonymous event depends on when it first appears in a string.
As a result, the verification of $K$- and infinite-step anonymity under the weak anonymous projection is much more difficult.
The classical observer cannot capture this scenario accurately.
Therefore, no existing methods for verifying $K$- and infinite-step opacity can be used to verify $K$- and infinite-step anonymity under the weak anonymous projection.

Using the designed concurrent-composition-based verification method, one accurately captures the state information to verify the $K$- and infinite-step anonymity under the weak anonymous projection.
Upper bounds on $K$ for $K$-step strong anonymity and weak anonymity are also obtained.
Furthermore, the proposed methodology is sufficiently general, as it can also be used to verify the $K$- and infinite-step opacity/anonymity under the natural projection efficiently.
A case study on an anonymous questioning process is provided for practical applications.

It is a fact that verifying $K$/infinite-step opacity in discrete-event systems is PSPACE-hard, as established in the literature \cite{Saboori2007,Saboori2012}. 
$K$/infinite-step anonymity \cite{Basilio2021} and $K$/infinite-step opacity are closely related information-flow properties. 
Their verification structures are essentially the same. 
This means that verifying $K$/infinite-step anonymity \cite{Basilio2021} in discrete-event systems is also PSPACE-hard.
As discussed in Remark \ref{re:3.2}, the proposed four anonymity notions reduce to the existing $K$-step and infinite-step anonymity \cite{Basilio2021} when $\Gamma=\emptyset$.
As a result, the verification of the four notions of anonymity proposed in this technical note is also PSPACE-hard.
It seems highly unlikely that there exist polynomial-time algorithms to verify the four types of anonymity considered in this paper unless P=PSPACE.
This means that when handling large-scale practical systems, the proposed algorithms will face challenges such as state space explosion problem (i.e., the state spaces of the constructed verifiers are all exponential) and memory consumption (i.e., the memory usage grows exponentially since the ``Breadth-First Search” traversal requires storing all visited states).
To overcome such limitations, in the future, we will consider employing approximate methods, such as abstraction-refinement \cite{Wang2021} and machine learning \cite{Ding2024}.
The former constructs a series of progressively precise approximate models to approach the verification problems of complex systems in a controllable manner, thereby achieving computational feasibility while ensuring correctness.
The latter provides probabilistic predictions whose accuracy depends on the quality of training data and model capabilities.
Both methods align with the regular structure of the constructed concurrent composition.
In addition, the idea of over-approximation state estimation proposed in \cite{Xiang2025} will also be considered.
To move forward, we will focus on enforcing these four types of anonymity based on the proposed verification method in the future.

\end{document}